\documentclass[11pt,a4paper]{article}
\usepackage[utf8]{inputenc}
\usepackage{jheppub}
\hypersetup{unicode=true,bookmarksopen=true}

\usepackage{todonotes}

\usepackage{mathtools}
\usepackage{lmodern}
\usepackage[T1]{fontenc}
\usepackage{bbm}
\DeclareMathAlphabet{\mathbfi}{OML}{cmm}{b}{it}

\usepackage{amsthm}
\newtheorem{theorem}{Theorem}
\newtheorem{lemma}[theorem]{Lemma}
\newtheorem{corollary}[theorem]{Corollary}
\newtheorem{definition}[theorem]{Definition}
\newtheorem{remark}[theorem]{Remark}

\let\originalleft\left
\let\originalright\right
\renewcommand{\left}{\mathopen{}\mathclose\bgroup\originalleft}
\renewcommand{\right}{\aftergroup\egroup\originalright}

\makeatletter
\newenvironment{equations}[1][]{\subequations\ifx\relax#1\relax\else\label{#1}\fi\align\ignorespaces}{\endalign\ignorespacesafterend\endsubequations}
\def\@spliteq#1{\begin{equation}\begin{split}#1\end{split}\end{equation}}
\def\splitequation{\collect@body\@spliteq}

\makeatother

\renewcommand{\vec}[1]{{\ifnum9<1#1\mathbf{#1}\else\ifcat\noexpand#1\relax\boldsymbol{#1}\else\mathbfi{#1}\fi\fi}}
\newcommand{\mathe}{\mathrm{e}}
\newcommand{\mathi}{\mathrm{i}}
\let\oldre\Re
\let\oldim\Im
\renewcommand{\Re}{\oldre\mathfrak{e}\,}
\renewcommand{\Im}{\oldim\mathfrak{m}\,}
\newcommand{\total}{\mathop{}\!\mathrm{d}}

\makeatletter
\def\abs{\@ifnextchar[\abs@size\abs@nosize}
\def\abs@size[#1]#2{\mathopen{#1\lvert}{#2}\mathclose{#1\rvert}}
\def\abs@nosize#1{\left\lvert{#1}\right\rvert}
\def\norm{\@ifnextchar[\norm@size\norm@nosize}
\def\norm@size[#1]#2{\mathopen{#1\lVert}{#2}\mathclose{#1\rVert}}
\def\norm@nosize#1{\left\lVert{#1}\right\rVert}
\makeatother

\newcommand{\sgn}{\operatorname{sgn}}

\newcommand{\1}{\mathbbm{1}}
\newcommand{\tr}{\operatorname{tr}}
\newcommand{\eqend}[1]{\,#1}
\newcommand{\bigo}[1]{\mathcal{O}\left({#1}\right)}

\makeatletter
\def\bra{\@ifnextchar[\bra@size\bra@nosize}
\def\bra@size[#1]#2{#1\langle{#2}#1\rvert}
\def\bra@nosize#1{\left\langle{#1}\right\rvert}
\def\ket{\@ifnextchar[\ket@size\ket@nosize}
\def\ket@size[#1]#2{#1\lvert{#2}#1\rangle}
\def\ket@nosize#1{\left\lvert{#1}\right\rangle}
\makeatother

\newcommand{\supp}{\operatorname{supp}}

\makeatletter
\def\dom{\@ifnextchar[\dom@size\dom@nosize}
\def\dom@size[#1]#2{\operatorname{dom}\mathopen{#1(}{#2}\mathclose{#1)}}
\def\dom@nosize#1{\operatorname{dom}\left({#1}\right)}
\makeatother

\makeatletter
\gdef\@fpheader{\strut}
\renewcommand\afterLogoSpace{\vskip-5pt}
\renewcommand\beforetochook{\pagestyle{myplain}\pagenumbering{roman}\afterRuleSpace}
\makeatother

\allowdisplaybreaks

\begin{document}

\title{Bounding relative entropy for non-unitary excitations in quantum field theory}

\author[a]{Markus B. Fr{\"o}b}
\author[b]{and Leonardo Sangaletti}

\affiliation[a]{Department Mathematik, Friedrich-Alexander-Universit{\"a}t Erlangen-N{\"u}rnberg, Cauerstra{\ss}e 11, 91058 Erlangen, Germany}
\affiliation[b]{Dipartimento di Fisica, Universit{\`a} di Genova, Via Dodecaneso 33, 16146 Genova, Italy}

\emailAdd{markus.froeb@fau.de}
\emailAdd{leonardo.sangaletti@edu.unige.it}

\abstract{We show how one can use the convexity of non-commutative $L^p$ norms to bound the relative entropy between a faithful state on a von Neumann algebra and an arbitrary excitation thereof. Our results hold for general von Neumann algebras, including the local algebras of type III that are ubiquitous in quantum field theory, and do not require knowledge of the relative modular operator. As an application of our results, we prove that for the chiral current on a light ray, the relative entropy between the vacuum and a dense set of single-particle states is uniformly bounded.}


\toccontinuoustrue
\maketitle

\section{Introduction}

Relative entropy, also known as Kullback--Leibler divergence, is a fundamental concept in classical information theory that quantifies the difference between two probability distributions. Its significance stems from its operational interpretation in hypothesis testing described by Stein's Lemma: the probability for an error of type II in a likelihood ratio test decreases exponentially with the number of measurements, with the convergence rate given by the relative entropy. In non-commutative probability theory, the relative entropy was first studied by Umegaki~\cite{umegaki1962}, and extended by Araki~\cite{araki1975,araki1976} to general von Neumann algebras using Tomita--Takesaki modular theory. At the same time, Uhlmann~\cite{uhlmann1977} introduced relative entropy for general positive functions of $*$-algebras using interpolation theory.

Consider a von Neumann algebra $\mathfrak{M}$ in standard form~\cite{haagerup1975} $(\mathfrak{M},\mathcal{H},J,\mathcal{P})$, with $\mathcal{H}$ the separable Hilbert space on which $\mathfrak{M}$ acts, $J$ a conjugate linear, isometric involution of $\mathcal{H}$ mapping $\mathfrak{M}$ into its commutant $\mathfrak{M}'$, and $\mathcal{P}$ a self-dual cone (invariant under $J$). Since $\mathcal{H}$ is separable, $\mathfrak{M}$ is $\sigma$-finite and admits a cyclic and separating vector $\Omega \in \mathcal{H}$~\cite[Ch. 2.5.1]{brattelirobinson1}, for which $J$ is the modular conjugation and $\mathcal{P}$ the associated natural positive cone $\mathcal{P} = \overline{ \{ a J a J \Omega \colon a \in \mathfrak{M} \}}$~\cite[Remark~1.2]{haagerup1975}. The relative entropy $S_\mathrm{rel}$\footnote{Here and in the following, we use a single vertical bar to denote the relative entropy between states, and a double bar to denote the relative entropy between their representative vectors.} between two faithful normal states $\psi$ and $\varphi$ is then given by the Araki formula~\cite{araki1975}
\begin{equation}
\label{eq:araki_relentropy}
S_\mathrm{rel}(\psi\vert\varphi) = S_\mathrm{rel}(\Psi^+\Vert\Phi^+) = - \left( \Phi^+, \ln \Delta_{\Psi^+,\Phi^+} \Phi^+ \right) \eqend{,}
\end{equation}
where $\Delta_{\Psi^+,\Phi^+}$ is the relative modular operator and $\Phi^+, \Psi^+ \in \mathcal{P}$ are the unique~\cite[Thm.~C.1 $(\delta 3)$]{arakimasuda1982} representative vectors of the normal states $\psi, \varphi$ in the natural positive cone. One of the main problems in computing the relative entropy consists in the derivation of an explicit expression for the relative modular operator. While in certain situations (for example the one covered by the Bisognano--Wichmann theorem~\cite{bisognanowichmann1975,bisognanowichmann1976}) the non-relative modular operator can be obtained, this is generally not the case for the relative one. An important exception corresponds to the case in which one considers the pair of vectors $\Omega$ and $U' U \Omega$ with unitaries $U' \in \mathfrak{M}'$ and $U \in \mathfrak{M}$, i.e., if one wants to compute the relative entropy between the cyclic and separating $\Omega$ and a unitary excitation $U' U \Omega$ thereof; other special cases include perturbations by an inner automorphism or where the perturbation is quadratic in the fields~\cite{brunettifredenhagenpinamonti2025}.

In this work, we want to derive bounds on the relative entropy between $\Omega$ and general (non-unitary) excitations of $\Omega$. Our main results are
\begin{theorem}
\label{thm:bound1}
For a von Neumann algebra $\mathfrak{M}$ with cyclic and separating vector $\Omega$, the relative entropy between $\Omega$ and $b' \Omega$ for any $b' \in \mathfrak{M}'$ with $\norm{ b' \Omega } = 1$ is bounded by
\begin{equation}
\label{eq:bound1_inequality}
0 \leq S_\mathrm{rel}\left( \Omega \Vert b' \Omega \right) \leq 2 \ln \norm{ \Delta^{- \frac{1}{4}} (b')^* b' \Omega } \leq 2 \ln \norm{ b' }_\mathrm{op} \eqend{,}
\end{equation}
where $\Delta$ is the modular operator associated to $\Omega$ and $\norm{ \cdot }_\mathrm{op}$ denotes the operator norm.

If $b'$ is only affiliated with $\mathfrak{M}'$ with $\Omega \in \dom{ (b')^* b' }$ and $\norm{ b' \Omega } = 1$, it holds that
\begin{equation}
\label{eq:bound1_inequality_finite}
0 \leq S_\mathrm{rel}\left( \Omega \Vert b' \Omega \right) \leq 2 \ln \norm{ \Delta^{- \frac{1}{4}} (b')^* b' \Omega } < \infty \eqend{,}
\end{equation}
and the relative entropy between $\Omega$ and $b' \Omega$ is finite.
\end{theorem}
This theorem translates the results of Jen{\v c}ov{\'a}~\cite{jencova2018,jencova2021} and Berta, Scholz and Toma\-michel~\cite{bertascholztomamichel2018} into the setting of general excitations of the given vector $\Omega$, and shows in particular that the relative entropy for such states is finite, even if the excitation comes from an unbounded operator. As a corollary, we obtain
\begin{corollary}
\label{corr:bound2}
In the setting of Thm.~\ref{thm:bound1}, consider the set of vectors $a \Omega$ with $\norm{ a \Omega } = 1$ for $a \in \mathfrak{M}$ analytic with respect to the modular flow, which is dense in $\mathcal{H}$. The relative entropy between $\Omega$ and $a \Omega$ is then bounded by
\begin{equation}
0 \leq S_\mathrm{rel}\left( \Omega \Vert a \Omega \right) \leq 2 \ln \norm{ \Delta^{- \frac{1}{4}} a \Delta a^* \Omega } = 2 \ln \norm{ \sigma_\frac{\mathi}{4}(a) \left[ \sigma_\frac{3 \mathi}{4}(a) \right]^* \Omega } \leq 2 \ln \norm{ \sigma_\frac{\mathi}{2}(a) }_\mathrm{op} \eqend{,}
\end{equation}
where $\sigma_t(a) \coloneq \Delta^{\mathi t} a \Delta^{- \mathi t}$ is the modular flow. The middle inequality is non-trivial, and the relative entropy is finite.
\end{corollary}
Even though the set of analytic vectors is dense, the relative entropy is not continuous in the states and thus this bound is not useful for general states. However, using the lower semicontinuity of relative entropy we obtain
\begin{corollary}
\label{corr:bound3}
In the setting of Thm.~\ref{thm:bound1}, consider a vector $a \Omega$ for a general $a \in \mathfrak{M}$. The relative entropy between $\Omega$ and $a \Omega$ is then bounded by
\begin{equation}
0 \leq S_\mathrm{rel}\left( \Omega \Vert a \Omega \right) \leq 2 \liminf_{n \to \infty} \ln \norm{ \Delta^{- \frac{1}{4}} a_n \Delta a_n^* \Omega } = 2 \liminf_{n \to \infty} \ln \norm{ \sigma_\frac{\mathi}{4}(a_n) \left[ \sigma_\frac{3 \mathi}{4}(a_n) \right]^* \Omega } \eqend{,}
\end{equation}
where
\begin{equation}
a_n \coloneq \sqrt{ \frac{n}{\pi} } \int_{-\infty}^\infty \mathe^{- n s^2} \sigma_s(a) \total s
\end{equation}
is analytic with respect to the modular flow $\sigma_t$ and satisfies
\begin{equation}
\lim_{n \to \infty} \norm{ a_n \Omega - a \Omega } = 0 \eqend{.}
\end{equation}
The bound may be trivial.
\end{corollary}

Bounds on (relative) entropy have also been studied previously, with perhaps the most famous conjecture given by the Bekenstein bound~\cite{bekenstein1981} which asserts that the entropy of any gravitating system has an upper bound given by $2 \pi R$ times its energy, with $R$ being the longest dimension of the system. The Bekenstein conjecture has been generalized in various ways, and at the same time there has been a long controversy about the exact definition of the various quantities that enter the conjecture (see, e.g., Refs.~\cite{page1982,sorkin1986,thooft1993,bousso1999,flanaganmarolfwald2000,marolfminicross2004,pesci2007,page2008,wall2011,blancocasinihungmyers2013,blancocasini2013}). Casini~\cite{casini2008} argued that the proper interpretation of the Bekenstein bound results in its equality to the positivity of relative entropy (which is a well-known result), and recently Longo~\cite{longo2025} showed that the Bekenstein bound in fact holds for the relative entropy itself. That is, one has
\begin{equation}
0 \leq S_\mathrm{rel}\left( \Phi \Vert \Omega \right) \leq 2 \pi R \left( \Phi, H \Phi \right)
\end{equation}
for a vector state $\Phi$ localized in a spacetime region of width $2 R$, i.e., such that $\left( \Phi, a \Phi \right) = \left( \Omega, a \Omega \right)$ for all $a$ supported outside the causal hull of this spacetime region, and where $H$ is the usual Hamiltonian. The bound can be generalized (with some error margin) also to states which are not strictly localized~\cite{hollandslongo2025}; see also Refs.~\cite{longoxu2018,kudlerflametal2025} for related results.

The relative entropy can be obtained as the limit $\alpha \to 1$ of the family of Petz--Rényi relative entropies~\cite{petz1985,petz1986a}, which may be defined for $\alpha \in (0,1)$ by~\cite{froebsangaletti2025}
\begin{equation}
\label{eq:petzrenyi_entropy}
S_\alpha(\Psi^+ \Vert \Phi^+) = \frac{1}{\alpha-1} \ln \left( \Phi^+, \Delta_{\Psi^+,\Phi^+}^{1-\alpha} \, \Phi^+ \right) \eqend{.}
\end{equation}
This family is monotonic in $\alpha$, such that the relative entropy gives an upper bound for the Petz--Rényi relative entropies with $\alpha < 1$. Moreover, Jen{\v c}ov{\'a} and Berta, Scholz and Tomamichel~\cite{jencova2018,bertascholztomamichel2018} have given a definition of sandwiched Rényi relative entropies for $\alpha \in (1,\infty)$, which employs non-commutative $L^p$ norms defined for von Neumann algebras (see also~\cite{hiai2019,kato2024} for further generalizations and properties). They have shown monotonicity for all $\alpha$, such that also the relative entropy can be bounded from above. In particular, the bounds of Thm.~\ref{thm:bound1} correspond to the value $\alpha = 2$ and the limit $\alpha \to \infty$, which are the cases where the non-commutative $L^p$ norms have simple explicit expressions. The remainder of this work is organized as follows: In Sec.~\ref{sec:lpnorms}, we recall the definition of (Araki--Masuda) non-commutative $L^p$ norms, and present a streamlined proof of monotonicity, which was given in~\cite{jencova2018,bertascholztomamichel2018}. In Sec.~\ref{sec:bounds}, we explicitly evaluate the non-commutative $L^p$ norms for $p = 4$ and $p = \infty$, which give the bounds of Thm.~\ref{thm:bound1}, and prove Cor.~\ref{corr:bound2} and Cor.~\ref{corr:bound3}. Lastly, in Sec.~\ref{sec:example} we consider the example of free scalar fields in a wedge and the chiral current on a light ray, and compute a bound on the relative entropy between the Minkowski vacuum and a single-particle state.

\section{Non-commutative \texorpdfstring{$L^p$}{Lp} norms}
\label{sec:lpnorms}

\subsection{Review of modular theory}

Let $\mathfrak{A}$ be a unital $\sigma$-finite von Neumann algebra and $\omega$ a faithful normal state $\omega$ on $\mathfrak{A}$ (such a state exists thanks to the $\sigma$-finite condition). Via the GNS construction, $\mathfrak{A}$ is represented by a von Neumann algebra of operators $\mathfrak{M} = \pi(\mathfrak{A})$ on an Hilbert space $\mathcal{H}$ with a cyclic (by GNS) and separating (since $\omega$ is faithful) vector $\Omega \in \mathcal{H}$ that implements the state $\omega$. We denote with $\Delta$ and $J$ the modular operator and the modular conjugation for the pair $(\mathfrak{M},\Omega)$. The natural positive cone $\mathcal{P}$ is defined as~\cite{araki1974}
\begin{equation}
\label{eq:pos_cone_1}
\mathcal{P} \coloneq \overline{ \{ \Delta^\frac{1}{4} a \Omega \colon a \in \mathfrak{M}, a \geq 0 \} } \eqend{,}
\end{equation}
is self-dual, satisfies $J \xi = \xi$ for all $\xi \in \mathcal{P}$, and coincides with the set~\cite[Thm.~4]{araki1974}
\begin{equation}
\label{eq:pos_cone_2}
\mathcal{P} = \overline{ \{ a J a J \Omega \colon a \in \mathfrak{M} \} } \eqend{.}
\end{equation}
The tuple $(\mathfrak{M},\mathcal{H},J,\mathcal{P})$ is the standard form~\cite{haagerup1975} of the von Neumann algebra.

Given two normal states $\varphi$, $\psi$ and their representative vectors $\Phi$, $\Psi$, the relative Tomita operator $S_{\Phi,\Psi}$ is defined by~\cite{araki1977}, \cite[Eq.~(C.1)]{arakimasuda1982}
\begin{equation}
\label{eq:relative_tomita_def}
S_{\Phi,\Psi} \left[ a \Psi + \left( 1 - s^{\mathfrak{M}'}(\Psi) \right) \zeta \right] = s^\mathfrak{M}(\Psi) a^* \Phi
\end{equation}
for all $a \in \mathfrak{M}$ and $\zeta \in \mathcal{H}$, where $s^\mathfrak{M}(\Psi)$ denotes the support projection of $\Psi$ in $\mathfrak{M}$ (the smallest projection in $\mathfrak{M}$ such that $s^\mathfrak{M}(\Psi) \Psi = \Psi$). $S_{\Phi,\Psi}$ is antilinear, closable and densely defined with domain $\dom{ S_{\Phi,\Psi} } = \mathfrak{M} \Psi + \bigl( \1 - s^{\mathfrak{M}'}(\Psi) \bigr) \mathcal{H}$, and its closure and adjoint have the support~$s(\overline{S}_{\Phi,\Psi}) = s^\mathfrak{M}(\Phi) s^{\mathfrak{M}'}(\Psi)$, $s( S_{\Phi,\Psi}^* ) = s^\mathfrak{M}(\Psi) s^{\mathfrak{M}'}(\Phi)$~\cite[Thm.~C.1 ($\alpha$)]{arakimasuda1982}. The closure admits the unique polar decomposition
\begin{equation}
\overline{S}_{\Phi,\Psi} = J_{\Phi,\Psi} \Delta_{\Phi,\Psi}^\frac{1}{2}
\end{equation}
into the modular conjugation $J_{\Phi,\Psi}$ and the square root of the positive relative modular operator $\Delta_{\Phi,\Psi} \coloneq ( S_{\Phi,\Psi} )^* \overline{S}_{\Phi,\Psi}$ with support $s( \Delta_{\Phi,\Psi} ) = s^\mathfrak{M}(\Phi) s^{\mathfrak{M}'}(\Psi)$~\cite[Thm.~C.1 ($\beta$)]{arakimasuda1982}. If $\Phi^+, \Psi^+ \in \mathcal{P}$ are the unique~\cite[Thm.~C.1 $(\delta 3)$]{arakimasuda1982} representative vectors in the natural positive cone, $J_{\Phi^+,\Psi^+} = J$~\cite[Thm.~2.4]{araki1977} or~\cite[Thm.~C.1 ($\epsilon1$)]{arakimasuda1982}. In general, we have
\begin{theorem}
\label{thm:modular_representative_change}
If $\Phi$ and $\Psi$ are representative vectors for $\varphi$ and $\psi$, there exist partial isometries $u',v' \in \mathfrak{M}'$ such that $\Phi = u' \Phi^+$, $\Psi = v' \Psi^+$, and
\begin{equation}
\label{eq:relative_tomita_representative_change}
\overline{S}_{\Phi,\Psi} = u' \overline{S}_{\Phi^+,\Psi^+} (v')^* \eqend{,}
\end{equation}
where $\Phi^+$ and $\Psi^+$ are the unique representative vectors for $\varphi$ and $\psi$ in the positive cone, as well as
\begin{equation}
\label{eq:relative_modular_representative_change}
\Delta_{\Phi,\Psi} = v' \Delta_{\Phi^+,\Psi^+} (v')^*
\end{equation}
and
\begin{equation}
\label{eq:relative_conjugation_representative_change}
J_{\Phi,\Psi} = u' J (v')^* \eqend{.}
\end{equation}
\end{theorem}
\begin{proof}
We follow parts of the proof of~\cite[Thm.~C.1]{arakimasuda1982}. By~\cite[Thm.~7 (5)]{araki1974} or~\cite[Thm.~C.1 ($\delta_2$)]{arakimasuda1982}, there exists a partial isometry $u' \in \mathfrak{M}'$ such that $\Phi = u' \Phi^+$ with $\Phi^+ \in \mathcal{P}$, satisfying $u' (u')^* = s^{\mathfrak{M}'}(\Phi)$ and $(u')^* u' = s^{\mathfrak{M}'}(\Phi^+)$ such that $\Phi^+ = (u')^* \Phi$, and analogously a partial isometry $v' \in \mathfrak{M}'$ such that $\Psi = v' \Psi^+$ with the same properties. From the equation~\eqref{eq:relative_tomita_def} for the relative modular operator, we thus obtain
\begin{splitequation}
&\left( S_{\Phi,\Psi} - u' S_{\Phi^+,\Psi^+} (v')^* \right) \left[ a \Psi + \left( 1 - s^{\mathfrak{M}'}(\Psi) \right) \zeta \right] \\
&= s^\mathfrak{M}(\Psi) a^* \Phi - u' S_{\Phi^+,\Psi^+} a \Psi^+ - u' S_{\Phi^+,\Psi^+} (v')^* \left( 1 - s^{\mathfrak{M}'}(\Psi) \right) \zeta \\
&= s^\mathfrak{M}(\Psi) a^* \Phi - u' s^\mathfrak{M}(\Psi^+) a^* \Phi^+ - u' S_{\Phi^+,\Psi^+} \left( 1 - s^{\mathfrak{M}'}(\Psi^+) \right) (v')^* \zeta \\
&= s^\mathfrak{M}(\Psi) a^* \Phi - s^\mathfrak{M}(\Psi^+) a^* u' \Phi^+ = \left[ s^\mathfrak{M}(\Psi) - s^\mathfrak{M}(\Psi^+) \right] a^* \Phi = 0 \eqend{,}
\end{splitequation}
where in the second equality we used that $(v')^* \left( 1 - s^{\mathfrak{M}'}(\Psi) \right) = (v')^* \left( 1 - v' (v')^* \right) =\linebreak \left( 1 - s^{\mathfrak{M}'}(\Psi^+) \right) (v')^*$, which is annihilated by $S_{\Phi^+,\Psi^+}$, and in the last that $s^\mathfrak{M}(\Psi) = s^\mathfrak{M}(\Psi^+)$ since $\Psi$ and $\Psi^+$ induce the same state on $\mathfrak{M}$. Thanks to the relation~\cite[Eq~(C.23)]{arakimasuda1982}, this extends to the closure and the adjoint, and for the relative modular operator it follows that
\begin{equation}
\Delta_{\Phi,\Psi} = v' S_{\Phi^+,\Psi^+}^* s^{\mathfrak{M}'}(\Phi^+) \overline{S}_{\Phi^+,\Psi^+} (v')^* = v' \Delta_{\Phi^+,\Psi^+} (v')^* \eqend{,}
\end{equation}
using that $s( S_{\Phi,\Psi}^* ) = s^\mathfrak{M}(\Psi) s^{\mathfrak{M}'}(\Phi)$~\cite[Thm.~C.1~$(\alpha)$]{arakimasuda1982}. Using furthermore that $(v')^* v' = s^{\mathfrak{M}'}(\Psi^+)$ and that $\Delta_{\Phi^+,\Psi^+}$ is supported on $s^\mathfrak{M}(\Phi^+) s^{\mathfrak{M}'}(\Psi^+)$, it follows that also $f(\Delta_{\Phi,\Psi}) = v' f(\Delta_{\Phi^+,\Psi^+}) (v')^*$. Finally, we obtain
\begin{equation}
J_{\Phi,\Psi} = \overline{S}_{\Phi,\Psi} \Delta_{\Phi,\Psi}^{-\frac{1}{2}} = u' \overline{S}_{\Phi^+,\Psi^+} s^{\mathfrak{M}'}(\Psi^+) \Delta_{\Phi^+,\Psi^+}^{-\frac{1}{2}} (v')^* = u' J (v')^* \eqend{.}
\end{equation}
\end{proof}

In particular, the relative modular operator $\Delta_{\Phi,\Psi}$ is independent of the choice of representative vector for $\varphi$. The relative entropy between the two states is then defined as~\cite[Def.~3.1]{araki1977}
\begin{equation}
\label{eq:relative_entropy_general}
S_\mathrm{rel}(\psi\vert\varphi) \coloneq \begin{cases} \displaystyle\int_0^\infty \ln \lambda \total \left( \Phi^+, E_\lambda^{\Phi^+,\Psi^+} \Phi^+ \right) & \text{if} \quad s(\psi) \geq s(\phi) \\ + \infty & \text{otherwise} \eqend{,} \end{cases}
\end{equation}
where $s(\omega)$ denotes the support projection of $\omega$ (the smallest projection $p$ in $\mathfrak{M}$ such that $\omega(p) = 1$) and $E_\lambda^{\Phi^+,\Psi^+}$ denotes the spectral projections of $\Delta_{\Phi^+,\Psi^+}$. If $s(\psi) \geq s(\phi)$ (and thus in particular when $\psi$ is faithful such that $s(\psi) = \1$), by using the conjugation property $\left[ \ln \Delta_{\Phi^+,\Psi^+} + J \ln \Delta_{\Psi^+,\Phi^+} J \right] s^\mathfrak{M}(\Phi^+) s^{\mathfrak{M}'}(\Psi^+) = 0$~\cite[Remark~3.4]{araki1977} one obtains
\begin{splitequation}
\label{eq:phi_elambda_conjugation}
0 &= \left( \Phi^+, \left[ E_\lambda^{\Phi^+,\Psi^+} + J E_\lambda^{\Psi^+,\Phi^+} J \right] s^\mathfrak{M}(\Phi^+) s^{\mathfrak{M}'}(\Psi^+) \Phi^+ \right) \\
&= \left( \Phi^+, E_\lambda^{\Phi^+,\Psi^+} s^{\mathfrak{M}'}(\Psi^+) J \Phi^+ \right) + \left( \Phi^+, J E_\lambda^{\Psi^+,\Phi^+} J s^{\mathfrak{M}'}(\Psi^+) J \Phi^+ \right) \\
&= \left( \Phi^+, E_\lambda^{\Phi^+,\Psi^+} J s^\mathfrak{M}(\Psi^+) \Phi^+ \right) + \left( \Phi^+, J E_\lambda^{\Psi^+,\Phi^+} s^\mathfrak{M}(\Psi^+) \Phi^+ \right) \\
&= \left( \Phi^+, E_\lambda^{\Phi^+,\Psi^+} J \Phi^+ \right) + \left( E_\lambda^{\Psi^+,\Phi^+} \Phi^+, J \Phi^+ \right) \\
&= \left( \Phi^+, E_\lambda^{\Phi^+,\Psi^+} \Phi^+ \right) + \left( \Phi^+, E_\lambda^{\Psi^+,\Phi^+} \Phi^+, \right) \eqend{,}
\end{splitequation}
where we also used the fact that $J \Phi^+ = \Phi^+$ since $\Phi^+ \in \mathcal{P}$, that $J$ is antilinear, that $J s^{\mathfrak{M}'}(\Psi^+) J = s^\mathfrak{M}(\Psi^+)$, and that $s^\mathfrak{M}(\Psi^+) \Phi^+ = \Phi^+$ since $s(\psi) \geq s(\phi)$. Therefore, we obtain
\begin{equation}
S_\mathrm{rel}(\psi\vert\varphi) = \begin{cases} - \displaystyle\int_0^\infty \ln \lambda \total \left( \Phi^+, E_\lambda^{\Psi^+,\Phi^+} \Phi^+ \right) & \text{if} \quad s(\psi) \geq s(\phi) \\ + \infty & \text{otherwise} \eqend{,} \end{cases}
\end{equation}
which coincides with the formula~\eqref{eq:araki_relentropy} quoted in the introduction. From Eq.~\eqref{eq:relative_modular_representative_change}, switching the roles of $\Phi$ and $\Psi$ we obtain
\begin{equation}
\left( \Phi, E_\lambda^{\Psi,\Phi} \Phi \right) = \left( u' \Phi^+, u' E_\lambda^{\Psi^+,\Phi^+} (u')^* u' \Phi^+ \right) = \left( \Phi^+, E_\lambda^{\Psi^+,\Phi^+} \Phi^+ \right) \eqend{,}
\end{equation}
and thus the relative entropy is independent of the choice of representative vectors and only depends on the states $\psi$ and $\varphi$. Analogously, one shows that the Petz--Rényi relative entropies~\eqref{eq:petzrenyi_entropy} are independent of the choice of representative vectors.

We also know how the modular objects are related when the states are exchanged, namely~\cite[Thm.~C.1 ($\beta$5), ($\epsilon$1)]{arakimasuda1982} (with the second one following from spectral calculus)
\begin{equation}
J_{\Phi,\Psi}^* = J_{\Psi,\Phi} \eqend{,} \quad J_{\Phi,\Psi} f(\Delta_{\Phi,\Psi}) J_{\Psi,\Phi} = f^*(\Delta_{\Psi,\Phi}^{-1}) \eqend{.}
\end{equation}
In particular, for representative vectors $\Phi^+, \Psi^+$ in the positive cone and for $s(\psi) \geq s(\phi)$ such that $s^{\mathfrak{M}}(\Psi^+) \Phi^+ = \Phi^+$, we have
\begin{equation}
\left( \Psi^+, \Delta_{\Phi^+,\Psi^+} f(\Delta_{\Phi^+,\Psi^+}) \Psi^+ \right) = \left( \Phi^+, f^*(\Delta_{\Psi^+,\Phi^+}^{-1}) \Phi^+ \right)
\end{equation}
by a computation analogous to equation~\eqref{eq:phi_elambda_conjugation}, using that
\begin{equation}
\Delta_{\Phi^+,\Psi^+}^\frac{1}{2} \Psi^+ = J s^\mathfrak{M}(\Psi^+) \Phi^+ = J \Phi^+ = \Phi^+ \eqend{.}
\end{equation}
That is, for $s(\psi) \geq s(\phi)$ it holds that
\begin{equation}
\int_0^\infty f(\lambda) \total \left( \Phi^+, E_\lambda^{\Psi^+,\Phi^+} \Phi^+ \right) = \int_0^\infty \lambda f^*(\lambda^{-1}) \total \left( \Psi^+, E_\lambda^{\Phi^+,\Psi^+} \Psi^+ \right) \eqend{.}
\end{equation}
From this, it follows that the relative entropy can alternatively be written as the Petz quasi-entropy~\cite{petz1985}
\begin{equation}
S_\mathrm{rel}(\psi\vert\varphi) = \begin{cases} \displaystyle \int_0^\infty \lambda \ln \lambda \total \left( \Psi^+, E_\lambda^{\Phi^+,\Psi^+} \Psi^+ \right) & \text{if} \quad s(\psi) \geq s(\phi) \\ + \infty & \text{otherwise} \eqend{,} \end{cases}
\end{equation}
which has the advantage that the integrand is continuous at $\lambda = 0$.

In fact, this is just one member of a family of divergences defined using the relative modular operator, the Petz quasi-entropies $Q_f$~\cite{petz1985}. Given a convex function $f \colon [0,\infty) \to \mathbb{R}$ and two normal states $\varphi, \psi$, they are defined by
\begin{equation}
Q_f\left( \varphi \vert \psi \right) \coloneq \begin{cases} \displaystyle \int_0^\infty f(\lambda) \total \left( \Psi, E_\lambda^{\Phi,\Psi} \Psi \right) & \text{if} \quad s(\psi) \geq s(\phi) \\ + \infty & \text{otherwise} \eqend{,} \end{cases}
\end{equation}
where $\Phi, \Psi$ are representative vectors of $\varphi, \psi$, and $E_\lambda^{\Phi,\Psi}$ denotes the spectral projections of $\Delta_{\Phi,\Psi}$. As for the relative entropy, one shows that the definition is independent of the choice of representative. For $f = \lambda \ln \lambda$, we recover the relative entropy $Q_f\left( \varphi \vert \psi \right) = S_\mathrm{rel}(\psi\vert\varphi)$ (note that the order of the states is reversed). By using the conjugation property of functions of the relative modular operator, we also obtain
\begin{equation}
\label{eq:conjrelQ}
Q_f\left( \varphi \vert \psi \right) = \begin{cases} \displaystyle \int_0^\infty \lambda f^*(\lambda^{-1}) \total \left( \Phi^+, E_\lambda^{\Psi^+,\Phi^+} \Phi^+ \right) & \text{if} \quad s(\psi) \geq s(\phi) \\ + \infty & \text{otherwise} \eqend{,} \end{cases}
\end{equation}
where now $\Phi^+,\Psi^+ \in \mathcal{P}$ are the unique representatives of $\varphi,\psi$ in the positive cone. Exchanging $\varphi$ and $\psi$, we see that $Q_f\left( \varphi \vert \psi \right) = Q_{\hat{f}}\left( \psi \vert \varphi \right)$ with the dual function $\hat{f}(\lambda) = \lambda f^*(\lambda^{-1})$ if $s(\psi) = s(\varphi)$, in particular if both quasi-entropies are finite.

It is then clear that if $f \leq g$, we have
\begin{equation}
Q_f\left( \varphi \vert \psi \right) \leq Q_g\left( \varphi \vert \psi \right) \eqend{.}
\end{equation}
This can be used to prove bounds on relative entropy. For example, since the function
\begin{equation}
\label{eq:fa_def}
f_a(\lambda) = \lambda \begin{cases} \frac{\lambda^a - 1}{a} & a \neq 0 \\ \ln \lambda & a = 0 \eqend{,} \end{cases}
\qquad f''_a(\lambda) = (a+1) \lambda^{a-1}
\end{equation}
is monotonically increasing in $a$ for all $\lambda \in [0,\infty)$ and convex for $a \geq -1$, also $Q_{f_a}$ is well-defined and monotonically increasing for $a \geq -1$ (but it might be infinite for some $a$, and thus all larger ones). In particular, using equation~\eqref{eq:conjrelQ} we have for $s(\psi) \geq s(\phi)$
\begin{splitequation}
\label{eq:quasientropy_limit_relative}
\lim_{a \to 0^-} Q_{f_a}\left( \varphi \vert \psi \right) &= \lim_{a \to 0^-} \int_0^\infty \frac{\lambda^{-a} - 1}{a} \total \left( \Phi^+, E_\lambda^{\Psi^+,\Phi^+} \Phi^+ \right) \\
&= - \int_0^\infty \ln \lambda \total \left( \Phi^+, E_\lambda^{\Psi^+,\Phi^+} \Phi^+ \right) = S_\mathrm{rel}\left( \psi \vert \varphi \right) \eqend{,}
\end{splitequation}
where the limit exists (but may be $+ \infty$) due to~\cite[Lemma~2]{petz1986b}. We note that the limit is essentially the definition of relative entropy given by Uhlmann~\cite{uhlmann1977}, and that Eq.~\eqref{eq:quasientropy_limit_relative} was shown by Araki~\cite[Eq.~(4.17)]{araki1976}. Furthermore, the same limit is obtained as $a \to 0^+$ if $\Phi \in \dom{ \Delta_{\Psi,\Phi}^{-\epsilon} }$ for some $\epsilon > 0$, again using~\cite[Lemma~2]{petz1986b}, and in this case is always finite.

The Petz--Rényi relative entropies~\eqref{eq:petzrenyi_entropy} can then be written as
\begin{splitequation}
\label{eq:petzrenyi_entropy_quasientropy}
S_\alpha(\psi \vert \varphi) &= \frac{1}{\alpha-1} \ln \int_0^\infty \lambda^{1-\alpha} \total \left( \Phi, E_\lambda^{\Psi,\Phi} \Phi \right) \\
&= \frac{1}{\alpha-1} \ln\left[ \left( \Phi, \Phi \right) + (\alpha-1) \int_0^\infty \frac{\lambda^{1-\alpha} - 1}{\alpha - 1} \total \left( \Phi, E_\lambda^{\Psi,\Phi} \Phi \right) \right] \\
&= \frac{1}{\alpha-1} \ln\left[ 1 + (\alpha-1) Q_{f_{\alpha-1}}(\varphi \vert \psi) \right] \eqend{,}
\end{splitequation}
and since $Q_{f_{\alpha-1}}$ is well-defined (but might be infinite) for all $\alpha \geq 0$, one can extend the definition of the Petz--Rényi relative entropies to $\alpha > 1$ by this formula. For $\alpha \to 1^+$ they have the same limit as the quasi-entropies $Q_{f_{\alpha-1}}$, namely the relative entropy:
\begin{theorem}
\label{thm:petzrenyi_limit}
If the Petz--Rényi divergence $S_\alpha(\psi \vert \varphi)$ is finite for some $\alpha = \alpha_0 > 1$, we have
\begin{equation}
\lim_{\alpha \to 1^+} S_\alpha(\psi \vert \varphi) = S_\mathrm{rel}(\psi \vert \varphi) < \infty \eqend{.}
\end{equation}
\end{theorem}
\begin{proof}
Since by assumption $S_{\alpha_0}(\psi \vert \varphi)$ is finite, it follows that also $Q_{f_{\alpha_0-1}}(\varphi \vert \psi) < \infty$, and since $Q_{f_a}(\varphi \vert \psi)$ is monotonically increasing in $a$ it follows that all $Q_{f_{\alpha-1}}(\varphi \vert \psi)$ with $\alpha \in [1,\alpha_0]$ are finite, and hence all $S_\alpha(\psi \vert \varphi)$ with $\alpha \in (1, \alpha_0]$ are finite. The fundamental inequality
\begin{equation}
\frac{x}{1+x} \leq \ln(1+x) \leq x
\end{equation}
then yields
\begin{equation}
\frac{Q_{f_{\alpha-1}}(\varphi \vert \psi)}{1+(\alpha-1) Q_{f_{\alpha-1}}(\varphi \vert \psi)} \leq S_\alpha(\psi \vert \varphi) \leq Q_{f_{\alpha-1}}(\varphi \vert \psi) \eqend{,}
\end{equation}
and since
\begin{equation}
\lim_{\alpha \to 1^+} (\alpha-1) Q_{f_{\alpha-1}}(\varphi \vert \psi) = \lim_{\alpha \to 1^+} \int_0^\infty \left( \lambda^{1-\alpha} - 1 \right) \total \left( \Phi^+, E_\lambda^{\Psi^+,\Phi^+} \Phi^+ \right) = 0 \eqend{,}
\end{equation}
where the last equality follows by dominated convergence, we obtain
\begin{equation}
\lim_{\alpha \to 1^+} S_\alpha(\psi \vert \varphi) = \lim_{\alpha \to 1^+} Q_{f_{\alpha-1}}(\varphi \vert \psi) = S_\mathrm{rel}(\psi \vert \varphi) \eqend{,}
\end{equation}
and the limit is finite.
\end{proof}

Finally, we want to give a relation between the relative modular operators for the algebra $\mathfrak{M}$ and for the commutant $\mathfrak{M}'$. From~\cite[Thm.~2.4]{araki1977} it follows that
\begin{equation}
\label{eq:delta_commutant}
\Delta'_{\Psi^+, \Phi^+} \Delta_{\Phi^+, \Psi^+} = \overline{F}^*_{\Psi^+, \Phi^+} \overline{F}_{\Psi^+, \Phi^+} \Delta_{\Phi^+, \Psi^+} = J \Delta_{\Psi^+, \Phi^+} J \Delta_{\Phi^+, \Psi^+} = s^{\mathfrak{M}'}(\Psi^+) s^\mathfrak{M}(\Phi^+) \eqend{,}
\end{equation}
i.e., the relative modular operator $\Delta'_{\Psi^+, \Phi^+}$ of the commutant is obtained by inverting and switching the two state vectors. By spectral calculus, this also holds for functions of $\Delta'_{\Psi^+, \Phi^+}$.

\subsection{Araki--Masuda divergences}

The non-commutative $L^p$ norms with $p \in [2,\infty]$ are defined for $\Psi \in \mathcal{H}$ by~\cite[Eq.~(1.4)]{arakimasuda1982}
\begin{equation}
\label{eq:lpnorm}
\norm{ \Psi }_{p,\Omega} \coloneq \sup_{\xi \in \mathcal{H}\colon \norm{\xi} = 1} \norm{ \Delta_{\xi, \Omega}^{\frac{1}{2}-\frac{1}{p}} \Psi } \eqend{,}
\end{equation}
where here and in the following we set $\frac{1}{\infty} \coloneq 0$, and $\norm{ \Delta_{\xi, \Omega}^{\frac{1}{2}-\frac{1}{p}} \Psi } \coloneq + \infty$ if $\Psi \not\in \dom{ \Delta_{\xi, \Omega}^{\frac{1}{2}-\frac{1}{p}} }$. Therefore, for $p > 2$ the norm $\norm{ \Psi }_{p,\Omega}$ may be infinite even if $\norm{ \Psi }$ is finite, since $\Psi$ might not be in the domain of the $(\frac{1}{2}-\frac{1}{p})$-th power of the relative modular operator $\Delta_{\xi, \Omega}$ for some or all $\xi$, or if the supremum is infinite. The norms for $p \in [1,2)$ may be defined either directly (by an analogous definition involving an infimum) or through a duality relation analogous to classical $L^p$ spaces. However, we will only need the norms for $p \geq 2$. Moreover, they only depend on the positive part of the polar decomposition of $\Psi$, i.e., for any $\Psi$ there exists a unique partial isometry $u \in \mathfrak{M}$ such that $u^* \Psi = \abs{\Psi} \in \mathcal{P}$ and $\norm{ \Psi }_{p,\Omega} = \norm{ \abs{\Psi} }_{p,\Omega}$~\cite[Thm.~3 (3)]{arakimasuda1982}. We have the fundamental inequality
\begin{lemma}[{\cite[Prop.~4]{bertascholztomamichel2018}}]
\label{lem:lpnorm_inequality}
The non-commutative $L^p$ norms~\eqref{eq:lpnorm} fulfill
\begin{equation}
\label{eq:lpnorm_interpolation}
\norm{ \Psi }_{p_\theta,\Omega} \leq \norm{ \Psi }_{p_0,\Omega}^{1-\theta} \norm{ \Psi }_{p_1,\Omega}^\theta \quad \text{with} \quad \frac{1}{p_\theta} \coloneq \frac{1-\theta}{p_0} + \frac{\theta}{p_1}
\end{equation}
for all $p_0, p_1 \in [2,\infty]$, $\theta \in [0,1]$. In particular, the function $p \mapsto \ln \norm{ \Psi }_{p,\Omega}^p$ is convex in $[2,\infty]$ for fixed $\Psi$ and $\Omega$~\cite[Cor.~5]{bertascholztomamichel2018}, that is
\begin{equation}
\label{eq:lpnorm_logconvex}
\ln \norm{ \Psi }_{p_\theta,\Omega}^{p_\theta} \leq \frac{p_\theta}{p_0} (1-\theta) \ln \norm{ \Psi }_{p_0,\Omega}^{p_0} + \frac{p_\theta}{p_1} \theta \ln \norm{ \Psi }_{p_1,\Omega}^{p_1} \eqend{.}
\end{equation}
\end{lemma}
\begin{proof}
The inequality~\eqref{eq:lpnorm_interpolation} is proven in~\cite{bertascholztomamichel2018} using interpolation theory, and we reproduce a streamlined version of the proof here. W.l.o.g. we may assume that $p_0 \leq p_1$ (otherwise change $\theta \to 1-\theta$), and that $\norm{ \Psi }_{p_1,\Omega} < \infty$ since otherwise the inequality is trivial. Note first that
\begin{equation}
\frac{1}{p_1} = \frac{1-\theta}{p_1} + \frac{\theta}{p_1} \leq \frac{1}{p_\theta} \leq \frac{1-\theta}{p_0} + \frac{\theta}{p_0} = \frac{1}{p_0} \eqend{,}
\end{equation}
and thus we have $p_\theta \geq 2$. From $\norm{ \Psi }_{p_1,\Omega} < \infty$ it follows by definition that $\Psi \in \dom{ \Delta_{\xi, \Omega}^{\frac{1}{2}-\frac{1}{p_1}} }$ for all $\xi \in \mathcal{H}$, and the inequality $\lambda^{1-\frac{2}{p}} \leq 1 + \lambda^{1 - \frac{2}{p_1}}$ for $\lambda \geq 0$ and $p \geq p_1$ shows that
\begin{equation}
\norm{ \Delta_{\xi, \Omega}^{\frac{1}{2}-\frac{1}{p}} \Psi }^2 \leq \norm{ \Psi }^2 + \norm{ \Delta_{\xi, \Omega}^{\frac{1}{2}-\frac{1}{p_1}} \Psi }^2 < \infty
\end{equation}
for all $\xi \in \mathcal{H}$ (by the spectral theorem, using that $\Delta_{\xi, \Omega} \geq 0$), such that taking the supremum over $\xi$ we have also $\norm{ \Psi }_{p_0,\Omega} < \infty$ and $\norm{ \Psi }_{p_\theta,\Omega} < \infty$.

Consider then the strip $B_\alpha \coloneq \{ z \in \mathbb{C} \colon \Re z \in [0,\alpha] \}$, and assume that $\Phi \in \dom{ \Delta_{\xi, \Omega}^\alpha }$. Then the function $z \mapsto \Delta_{\xi, \Omega}^z \Phi$ is holomorphic in the interior of $B_\alpha$ and continuous and bounded on its closure~\cite[Lemma~VI~2.3]{takesaki2003}, and it follows that the function $f(z) = \norm{ \Delta_{\xi, \Omega}^z \Phi }$ has the same properties. We apply Hadamard's three lines lemma
\begin{equation}
\label{eq:hadamard_lemma}
\abs{ f(x + \mathi y) } \leq \left[ \sup_{y \in \mathbb{R}} \abs{ f(a + \mathi y) } \right]^\frac{b - x}{b-a} \left[ \sup_{y \in \mathbb{R}} \abs{ f(b + \mathi y) } \right]^\frac{x - a}{b-a} \eqend{,} \quad x = \Re z \in [a, b]
\end{equation}
to obtain
\begin{equation}
\norm{ \Delta_{\xi, \Omega}^x \Phi } = \norm{ \Delta_{\xi, \Omega}^{x+\mathi y} \Phi } \leq \left[ \sup_{y \in \mathbb{R}} \norm{ \Delta_{\xi, \Omega}^{\mathi y} \Phi } \right]^{1-\frac{x}{\alpha}} \left[ \sup_{y \in \mathbb{R}} \norm{ \Delta_{\xi, \Omega}^{\alpha + \mathi y} \Phi } \right]^\frac{x}{\alpha} = \norm{ \Phi }^{1-\frac{x}{\alpha}} \norm{ \Delta_{\xi, \Omega}^\alpha \Phi }^\frac{x}{\alpha}
\end{equation}
for all $x \in [0,\alpha]$. Taking $\alpha = \frac{1}{p_0} - \frac{1}{p_1}$, $\Phi = \Delta_{\xi, \Omega}^{\frac{1}{2} - \frac{1}{p_0}} \Psi$ (which fulfills the assumption since $\Psi \in \dom{ \Delta_{\xi, \Omega}^{\frac{1}{2} - \frac{1}{p_1}} }$) and $x = \alpha \theta$, it follows that
\begin{equation}
\norm{ \Delta_{\xi, \Omega}^{\frac{1}{2} - \frac{1-\theta}{p_0} - \frac{\theta}{p_1}} \Psi } \leq \norm{ \Delta_{\xi, \Omega}^{\frac{1}{2} - \frac{1}{p_0}} \Psi }^{1-\theta} \norm{ \Delta_{\xi, \Omega}^{\frac{1}{2} - \frac{1}{p_1}} \Psi }^\theta
\end{equation}
for all $\theta \in [0,1]$. Taking the supremum over $\xi \in \mathcal{H}$ with $\norm{ \xi } = 1$, the first inequality~\eqref{eq:lpnorm_interpolation} follows, and the second one~\eqref{eq:lpnorm_logconvex} is obtained by a straightforward computation.
\end{proof}

We note that the non-commutative $L^p$ norms with $p \geq 2$ are finite for $\Psi = b \Omega$ with $b \in \mathfrak{M}$, since these are in the domain of $\Delta_{\xi, \Omega}^\frac{1}{2}$ and thus in the domain of $\Delta_{\xi, \Omega}^{\frac{1}{2} - \frac{1}{p}}$ for all $p \in [2,\infty]$. In fact, we have $\norm{ \Delta_{\xi,\Omega}^\frac{1}{2} b \Omega } = \norm{ J_{\xi,\Omega} \overline{S}_{\xi,\Omega} b \Omega } = \norm{ b^* \xi } \leq \norm{ b^* }_\mathrm{op}$, where we used that $\norm{ J_{\xi,\Omega} }_\mathrm{op} = 1$. However, they can also be extended to operators affiliated with $\mathfrak{M}$, which are the operators commuting with every unitary $u' \in \mathfrak{M}'$. In particular, consider a closed and densely defined operator $a$ affiliated with $\mathfrak{M}$ and its polar decomposition $a = u \abs{a}$. Then~\cite[Lemma 4.4.1]{murrayvonneumann1936} the partial isometry lies in the algebra, $u \in \mathfrak{M}$, and the absolute value $\abs{a}$ is affiliated with $\mathfrak{M}$. Let $E_I$ be the spectral projections of $\abs{a}$, then $E_{[0,n]} \abs{a}$ is bounded and hence an element of $\mathfrak{M}$, for which the non-commutative $L^p$ norm (when applied to $\Omega$) is finite.
\begin{definition}
\label{def:lpnorm_affiliated}
If the limit
\begin{equation}
\norm{ a \Omega }_{p,\Omega} \coloneq \lim_{n \to \infty} \norm{ u E_{[0,n]} \abs{a} \Omega }_{p,\Omega} = \lim_{n \to \infty} \norm{ E_{[0,n]} \abs{a} \Omega }_{p,\Omega}
\end{equation}
is finite, we call it the $L^p$ norm of $a \Omega$. Otherwise we set $\norm{ a \Omega }_{p,\Omega} \coloneq + \infty$.
\end{definition}
Therefore, even for $p \geq 2$ the $L^p$ norm for affiliated operators takes values in the extended positive reals $\overline{\mathbb{R}}_+$. The last equality is obtained as follows: Since $\norm{ \Psi }_{p,\Omega} = \norm{ \abs{\Psi} }_{p,\Omega}$~\cite[Thm.~3 (3)]{arakimasuda1982} where $\abs{\Psi} = u^* \Psi \in \mathcal{P}$ for a unique partial isometry $u \in \mathfrak{M}$ with $u^* u = s^\mathfrak{M}(\abs{\Psi})$ and $u u^* = s^\mathfrak{M}(\Psi)$, it follows that the non-commutative $L^p$ norms are invariant under the action of unitary operators $v \in \mathfrak{M}$. Namely, consider the vector $v \Psi = v u \abs{\Psi}$, and note that $v u$ is still a partial isometry. Since $( v u )^* v u = u^* u = s^\mathfrak{M}(\abs{\Psi})$ and $v u ( v u )^* = v s^\mathfrak{M}(\Psi) v^* = s^\mathfrak{M}(v \Psi)$, where the last equality follows because by definition $s^\mathfrak{M}(v \Psi)$ is the projector on $\overline{\mathfrak{M}' v \Psi} = v \overline{\mathfrak{M}' \Psi}$, by uniqueness we have $\abs{v \Psi} = \abs{\Psi}$. Because $\Psi$ and $v \Psi$ correspond to different states on $\mathfrak{M}$, (functions of) the non-commutative $L^p$ norms for the algebra are not good measures of distinguishability. On the other hand, since acting with a unitary operator $v' \in \mathfrak{M}'$ from the commutant does not change the state on the algebra, the non-commutative $L^p$ norms for the commutant seem more sensible.

For $\alpha \in (1,\infty)$, \cite{bertascholztomamichel2018} defines the Araki--Masuda divergence of order $\alpha$ as
\begin{equation}
\label{eq:arakimasuda_divergence}
D_\alpha\left( \omega \vert \psi \right) \coloneq \frac{1}{\alpha-1} \ln \norm{ \Psi }^{2 \alpha}_{2 \alpha,\Omega} \eqend{,}
\end{equation}
where $\Psi$ is a representative vector of the normal state $\psi$ (and thus has $\norm{ \Psi } = 1$), and the non-commutative $L^p$ norm is computed with respect to the commutant $\mathfrak{M}'$~\cite[App.~C]{bertascholztomamichel2018}, \cite[Sec.~3.1.1]{jencova2018}. From the above properties of the non-commutative $L^p$ norms, we obtain
\begin{lemma}[{\cite[Lemma 8]{bertascholztomamichel2018}}]
\label{lem:arakimasuda_increasing}
For fixed states $\omega$ and $\psi$, the Araki--Masuda divergence is a continuous and monotonically increasing function of $\alpha \in (1,\infty)$.
\end{lemma}
\begin{proof}
Eq.~\eqref{eq:lpnorm_logconvex} shows that $\ln \norm{ \Psi }^{2 \alpha}_{2 \alpha,\Omega}$ is convex for $\alpha \in [1,\infty]$, hence continuous for $\alpha \in (1,\infty)$. For $\alpha = 1$, we have
\begin{equation}
\ln \norm{ \Psi }^{2 \alpha}_{2 \alpha,\Omega} = 2 \ln \left[ \sup_{\xi \in \mathcal{H}\colon \norm{\xi} = 1} \norm{ \Psi } \right] = 2 \ln 1 = 0 \eqend{,}
\end{equation}
since $\Psi$ is normalized. Eq.~\eqref{eq:lpnorm_logconvex} with $p_0 = 2$, $p_1 = 2 \beta > 2 \alpha = p_\theta > 2$ and $\theta = \frac{\beta (\alpha-1)}{\alpha (\beta-1)} \in (0,1)$ then shows that
\begin{equation}
\ln \norm{ \Psi }_{2 \alpha,\Omega}^{2 \alpha} \leq \frac{\alpha - 1}{\beta - 1} \ln \norm{ \Psi }_{2 \beta,\Omega}^{2 \beta} \eqend{,}
\end{equation}
which is the monotonous increase.
\end{proof}

If we consider $\psi = \omega \circ \operatorname{ad}_{b'} = \left( b' \Omega, \cdot\, b' \Omega \right)$ with $b'$ affiliated with $\mathfrak{M}'$ and $\Omega \in \dom{b'}$, this defines a state if $\omega\left( (b')^* b' \right) = 1$. However, it is possible that the Araki--Masuda divergences become infinite for some $\alpha_0$ (and hence also all larger $\alpha \geq \alpha_0$). Nevertheless, Lemma~\ref{lem:arakimasuda_increasing} holds (with values in the extended positive reals $\overline{\mathbb{R}}_+$) for $\Psi = b' \Omega$. Namely, for each $n \in \mathbb{N}$, Lemma~\ref{lem:arakimasuda_increasing} shows that
\begin{equation}
\ln \norm{ u E_{[0,n]} \abs{b'} \Omega }_{2 \alpha,\Omega}^{2 \alpha} \leq \frac{\alpha - 1}{\beta - 1} \ln \norm{ u E_{[0,n]} \abs{b'} \Omega }_{2 \beta,\Omega}^{2 \beta} \eqend{,}
\end{equation}
and taking the limit $n \to \infty$, the inequality remains true, possibly in $\overline{\mathbb{R}}_+$.

In the limit $\alpha \to 1$, the Araki--Masuda divergence reduces to the relative entropy~\eqref{eq:relative_entropy_general}. This was proven in~\cite[Prop.~3.8]{jencova2018} and ~\cite[Thm.~13]{bertascholztomamichel2018} under certain conditions. Since this is important for our results, we provide again the (somewhat) simplified and streamlined proof.
\begin{theorem}[{\cite[Cor.~3.6]{jencova2018}, \cite[Thm.~12]{bertascholztomamichel2018}}]
\label{thm:arakimasuda_bounded_petzrenyi}
For all $\alpha > 1$, the Araki--Masuda divergence is bounded from above and below by the Petz--Rényi relative entropy according to
\begin{equation}
\label{eq:arakimasuda_bounded_petzrenyi}
S_{2-\frac{1}{\alpha}}\left( \omega \vert \psi \right) \leq D_\alpha\left( \omega \vert \psi \right) \leq S_\alpha\left( \omega \vert \psi \right) \eqend{.}
\end{equation}
\end{theorem}
\begin{proof}
Since $\omega$ is faithful such that $s(\omega) = s^\mathfrak{M}(\Omega) = \1$, the condition $s(\psi) \leq s(\omega)$ is always fulfilled. From the definitions~\eqref{eq:arakimasuda_divergence} of the Araki--Masuda divergence and~\eqref{eq:petzrenyi_entropy_quasientropy} of the Petz--Rényi relative entropy, which we write in the form
\begin{equation}
S_\alpha\left( \omega \vert \psi \right) = \frac{1}{\alpha-1} \ln \norm{ \Delta_{\Omega, \Psi}^\frac{1-\alpha}{2} \Psi }^2 \eqend{,}
\end{equation}
and using that the logarithm is monotonically increasing, we obtain that Eq.~\eqref{eq:arakimasuda_bounded_petzrenyi} is equivalent to
\begin{equation}
\label{eq:arakimasuda_proof_ineq}
\norm{ \Delta_{\Omega, \Psi}^\frac{1-\alpha}{2 \alpha} \Psi } \leq \norm{ \Psi }_{2 \alpha,\Omega} \leq \norm{ \Delta_{\Omega, \Psi}^\frac{1-\alpha}{2} \Psi }^\frac{1}{\alpha} \eqend{,}
\end{equation}
where we recall that the $L^p$ norm is computed with respect to the commutant $\mathfrak{M}'$. Writing $\Psi = u' \Psi^+$ with $u' \in \mathfrak{M}'$ a partial isometry and $\Psi^+ \in \mathcal{P}$ the unique representative in the positive cone, we have shown before that $S_\alpha\left( \omega \vert \psi \right)$ only depends on $\Psi^+$, i.e., that $\norm{ \Delta_{\Omega,\Psi}^r \Psi } = \norm{ \Delta_{\Omega,\Psi^+}^r \Psi^+ }$. Moreover, the Araki--Masuda divergences only depend on the positive part $\abs{\Psi}$ of the polar decomposition of $\Psi$, but since the $L^p$ norm is computed with respect to the commutant $\mathfrak{M}'$ we have $\Psi = u' \Psi^+ = v' \abs{\Psi}$ for another partial isometry $v' \in \mathfrak{M}'$. By the uniqueness of both decompositions and since $\abs{\Psi} \in \mathcal{P}$, we obtain $u' = v'$ and $\Psi^+ = \abs{\Psi}$. Note that $\Psi^+ = \abs{\Psi}$ is a representative of the state $\psi$ on $\mathfrak{M}$ (which is what we are interested in), but not of $\psi$ as a state on $\mathfrak{M}'$.

To show the inequalities~\eqref{eq:arakimasuda_proof_ineq} with $\Psi = \Psi^+$, we rewrite the outer terms by using that the relative modular operator $\Delta'$ of the commutant is related to the one of the algebra via~\eqref{eq:delta_commutant}
\begin{equation}
( \Delta'_{\Psi^+, \Omega} )^r \Delta_{\Omega, \Psi^+}^r = s^{\mathfrak{M}'}(\Psi^+) s^\mathfrak{M}(\Omega) = s^{\mathfrak{M}'}(\Psi^+) \eqend{.}
\end{equation}
We therefore obtain
\begin{equation}
\norm{ \Delta_{\Omega, \Psi}^r \Psi } = \norm{ \Delta_{\Omega, \Psi^+}^r \Psi^+ } = \norm{ ( \Delta'_{\Psi^+, \Omega} )^{-r} ( \Delta'_{\Psi^+, \Omega} )^r \Delta_{\Omega, \Psi^+}^r \Psi^+ } = \norm{ ( \Delta'_{\Psi^+, \Omega} )^{-r} \Psi^+ } \eqend{,}
\end{equation}
and it follows that we may write~\eqref{eq:arakimasuda_proof_ineq} as
\begin{equation}
\norm{ ( \Delta'_{\Psi^+, \Omega} )^{\frac{1}{2} - \frac{1}{2 \alpha}} \Psi^+ } \leq \sup_{\xi \in \mathcal{H}\colon \norm{\xi} = 1} \norm{ ( \Delta'_{\xi, \Omega} )^{\frac{1}{2}-\frac{1}{2 \alpha}} \Psi^+ } \leq \norm{ ( \Delta'_{\Psi^+, \Omega} )^\frac{\alpha-1}{2} \Psi^+ }^\frac{1}{\alpha} \eqend{.}
\end{equation}
The first inequality is now obvious, and for the second we write
\begin{splitequation}
\norm{ ( \Delta'_{\Psi^+, \Omega} )^\frac{\alpha-1}{2} \Psi^+ } &= \norm{ ( \Delta'_{\Psi^+, \Omega} )^\frac{\alpha}{2} ( \Delta'_{\Psi^+, \Omega} )^{-\frac{1}{2}} \Delta_{\Omega, \Psi^+}^{-\frac{1}{2}} \Delta_{\Omega, \Psi^+}^\frac{1}{2} \Psi^+ } \\
&= \norm{ ( \Delta'_{\Psi^+, \Omega} )^\frac{\alpha}{2} s^{\mathfrak{M}'}(\Psi^+) J \overline{S}_{\Omega, \Psi^+} \Psi^+ } = \norm{ ( \Delta'_{\Psi^+, \Omega} )^\frac{\alpha}{2} \Omega } \eqend{,}
\end{splitequation}
using that $s(\Delta'_{\Psi^+, \Omega}) = s^{\mathfrak{M}'}(\Psi^+) = J s^\mathfrak{M}(\Psi^+) J$, such that we have to show that
\begin{equation}
\label{eq:arakimasuda_proof_ineq_2}
\sup_{\xi \in \mathcal{H}\colon \norm{\xi} = 1} \norm{ ( \Delta'_{\xi, \Omega} )^{\frac{1}{2}-\frac{1}{2 \alpha}} \Psi^+ } \leq \norm{ ( \Delta'_{\Psi^+, \Omega} )^\frac{\alpha}{2} \Omega }^\frac{1}{\alpha} \eqend{.}
\end{equation}
Since the inequality is trivial if $\norm{ ( \Delta'_{\Psi^+, \Omega} )^\frac{\alpha}{2} \Omega } = \infty$, we may assume w.l.o.g. that $\Omega \in \dom{ ( \Delta'_{\Psi^+, \Omega} )^\frac{\alpha}{2} }$, and hence the function $z \mapsto ( \Delta'_{\Psi^+, \Omega} )^{\alpha z} \Omega$ is holomorphic in the interior of $B_\frac{1}{2}$ and continuous and bounded on its closure~\cite[Lemma~VI~2.3]{takesaki2003}. The same holds for any $\chi \in \dom{ ( \Delta'_{\xi, \Omega} )^\frac{1}{2} }$ for the function $z \mapsto ( \Delta'_{\xi, \Omega} )^{\frac{1}{2} - z} \chi$, such that the function $f_\chi$ with
\begin{equation}
f_\chi(z) = \left( ( \Delta'_{\xi, \Omega} )^{\frac{1}{2} - z^*} \chi, ( \Delta'_{\Psi^+, \Omega} )^{\alpha z} \Omega \right)
\end{equation}
is holomorphic in the interior of $B_\frac{1}{2}$ and continuous and bounded on its closure.

Let us show first that on the boundaries $\Re z = 0$ and $\Re z = \frac{1}{2}$ of $B_\frac{1}{2}$ the vector $( \Delta'_{\Psi^+, \Omega} )^{\alpha z} \Omega$ lies in the domain of $( \Delta'_{\xi, \Omega} )^{\frac{1}{2} - z}$. Since by Thm.~\ref{thm:modular_representative_change} the relative modular operator $\Delta'_{\xi, \Omega}$ does not depend on the representative vector $\xi$ but only on the state $\chi = \left( \xi, \cdot \, \xi \right)$ on $\mathfrak{M}'$, we may replace $\xi$ by $\xi^+ \in \mathcal{P}$, which is also the positive cone for $\mathfrak{M}'$. Analogously to the proof of~\cite[Lemma C.2 (1)]{arakimasuda1982}, we then obtain for $\Re z = 0$
\begin{splitequation}
\norm{ ( \Delta'_{\xi, \Omega} )^{\frac{1}{2} - \mathi y} ( \Delta'_{\Psi, \Omega} )^{\mathi \alpha y} \Omega } &= \norm{ ( \Delta'_{\xi^+, \Omega} )^{\frac{1}{2} - \mathi \alpha y} ( \Delta'_{\Psi, \Omega} )^{\mathi \alpha y} \Omega } \\
&= \norm{ J ( \Delta'_{\xi^+, \Omega} )^{\frac{1}{2} - \mathi \alpha y} ( \Delta'_{\Psi, \Omega} )^{\mathi \alpha y} \Omega } = \norm{ \overline{F}_{\xi^+, \Omega} ( D \chi \colon\! D \psi )'_{-\alpha y} \Omega } \\
&= \norm{ ( \Delta'_{\Psi, \Omega} )^{- \mathi \alpha y} ( \Delta'_{\xi^+, \Omega} )^{\mathi \alpha y} \xi^+ } = \norm{ \xi^+ } = 1 \eqend{,}
\end{splitequation}
where we used that $( \Delta'_{\xi^+, \Omega} )^{\mathi t}$ is unitary, that~\cite[Thm.~C.1 ($\beta$2) applied to the commutant]{arakimasuda1982}
\begin{equation}
( \Delta'_{\xi^+, \Omega} )^{- \mathi \alpha y} ( \Delta'_{\Psi, \Omega} )^{\mathi \alpha y} = ( D \chi \colon\! D \psi )'_{-\alpha y} \in \mathfrak{M}' \eqend{,}
\end{equation}
and that $J'_{\xi^+, \Omega} = J$~\cite[Thm.~2.4]{araki1977}. On the other hand, for $\Re z = \frac{1}{2}$ we have
\begin{equation}
\norm{ ( \Delta'_{\xi, \Omega} )^{- \mathi y} ( \Delta'_{\Psi^+, \Omega} )^{\frac{\alpha}{2} + \mathi \alpha y} \Omega } = \norm{ ( \Delta'_{\Psi^+, \Omega} )^\frac{\alpha}{2} \Omega } \eqend{,}
\end{equation}
which is finite by assumption. Hence by the definition of the adjoint for $\chi \in \dom{ ( \Delta'_{\xi, \Omega} )^\frac{1}{2} }$ we have the bounds
\begin{equations}
\abs{ f_\chi( \mathi y ) } &\leq \norm{ ( \Delta'_{\xi, \Omega} )^{\frac{1}{2} - \mathi y} ( \Delta'_{\Psi, \Omega} )^{\mathi \alpha y} \Omega } \norm{ \chi } \leq \norm{ \chi } \eqend{,} \\
\abs{ f_\chi\left( \frac{1}{2} + \mathi y \right) } &\leq \norm{ ( \Delta'_{\xi, \Omega} )^{- \mathi y} ( \Delta'_{\Psi, \Omega} )^{\frac{\alpha}{2} + \mathi \alpha y} \Omega } \norm{ \chi } \leq \norm{ ( \Delta'_{\Psi^+, \Omega} )^\frac{\alpha}{2} \Omega } \norm{ \chi } \eqend{,}
\end{equations}
and employing again Hadamard's three lines lemma~\eqref{eq:hadamard_lemma}, it follows that
\begin{equation}
\abs{ f_\chi(z) } \leq \norm{ ( \Delta'_{\Psi^+, \Omega} )^\frac{\alpha}{2} \Omega }^{2 \Re z} \norm{ \chi }
\end{equation}
for all $z \in B_\frac{1}{2}$. Since
\begin{equation}
\abs{ f_\chi(z) } = \abs{ [ f_\chi(z) ]^* } = \abs{ \left( ( \Delta'_{\Psi^+, \Omega} )^{\alpha z} \Omega, ( \Delta'_{\xi, \Omega} )^{\frac{1}{2} - z^*} \chi \right) } \eqend{,}
\end{equation}
it follows that $\chi \mapsto [ f_\chi(z) ]^*$ is a bounded linear map for all $z \in B_\frac{1}{2}$, such that $( \Delta'_{\Psi^+, \Omega} )^{\alpha z} \Omega$ is in the domain of $( \Delta'_{\xi, \Omega} )^{\frac{1}{2} - z}$ for all $z \in B_\frac{1}{2}$ and
\begin{equation}
f_\chi(z) = \left( \chi, ( \Delta'_{\xi, \Omega} )^{\frac{1}{2} - z} ( \Delta'_{\Psi^+, \Omega} )^{\alpha z} \Omega \right)
\end{equation}
holds. Moreover, since $\dom{ ( \Delta'_{\xi, \Omega} )^{\frac{1}{2} - z^*} } \supset \dom{ ( \Delta'_{\xi, \Omega} )^\frac{1}{2} }$ is dense in $\mathcal{H}$ we can extend the map $\chi \mapsto [ f_\chi(z) ]^*$ by continuity to all $\chi \in \mathcal{H}$ and obtain
\begin{equation}
\norm{ ( \Delta'_{\xi, \Omega} )^{\frac{1}{2} - z} ( \Delta'_{\Psi^+, \Omega} )^{\alpha z} \Omega } = \sup_{\chi \in \mathcal{H} \colon \norm{ \chi } = 1 } \abs{ f_\chi(z) } \leq \norm{ ( \Delta'_{\Psi^+, \Omega} )^\frac{\alpha}{2} \Omega }^{2 \Re z} \eqend{.}
\end{equation}
Choosing $z = (2\alpha)^{-1}$, it follows that
\begin{equation}
\norm{ ( \Delta'_{\xi, \Omega} )^{\frac{1}{2} - \frac{1}{2 \alpha}} ( \Delta'_{\Psi^+, \Omega} )^\frac{1}{2} \Omega } \leq \norm{ ( \Delta'_{\Psi^+, \Omega} )^\frac{\alpha}{2} \Omega }^\frac{1}{\alpha} \eqend{,}
\end{equation}
and taking the supremum over all $\xi \in \mathcal{H}$ with $\norm{ \xi } = 1$ and using that $( \Delta'_{\Psi^+, \Omega} )^\frac{1}{2} \Omega = J \Psi^+ = \Psi^+$, we also obtain the second inequality~\eqref{eq:arakimasuda_proof_ineq_2}.
\end{proof}
From this, we obtain
\begin{theorem}[{\cite[Prop.~3.8]{jencova2018}}]
\label{thm:arakimasuda_limit}
If the Araki--Masuda divergence $D_\alpha\left( \omega \vert \psi \right)$ is finite for some $\alpha = \alpha_0 > 1$, we have
\begin{equation}
\lim_{\alpha \to 1^+} D_\alpha\left( \omega \vert \psi \right) = S_\mathrm{rel}\left( \omega \vert \psi \right) < \infty \eqend{.}
\end{equation}
\end{theorem}
\begin{proof}
From Thm.~\ref{thm:arakimasuda_bounded_petzrenyi}, it follows that $\lim_{\alpha \to 1^+} D_\alpha\left( \omega \vert \psi \right) = \lim_{\alpha \to 1^+} S_\alpha\left( \omega \vert \psi \right)$. Since by assumption $D_{\alpha_0}\left( \omega \vert \psi \right) < \infty$ and $D_\alpha\left( \omega \vert \psi \right)$ is monotonically increasing in $\alpha$, we have $D_\alpha\left( \omega \vert \psi \right) < \infty$ for all $\alpha \in (1,\alpha_0]$. From the lower bound of the Araki--Masuda divergence in Thm.~\ref{thm:arakimasuda_bounded_petzrenyi}, it follows that $S_\alpha\left( \omega \vert \psi \right) < \infty$ for all $\alpha \in (1,2-\alpha_0^{-1}]$, and thus the limit exists and is finite by the corresponding result for the Petz--Rényi relative entropy, Thm.~\ref{thm:petzrenyi_limit}.
\end{proof}

\section{Bounds on relative entropy}
\label{sec:bounds}

By Thms.~\ref{thm:petzrenyi_limit} and~\ref{thm:arakimasuda_bounded_petzrenyi}, we can bound the relative entropy by the Araki--Masuda divergence for any $\alpha > 1$, assuming that it is finite for at least one $\alpha_0 > 1$. In general, we obtain the bound
\begin{equation}
\label{eq:srel_bound_dalpha}
S_\mathrm{rel}\left( \omega \vert \psi \right) \leq S_{2 - \frac{1}{\alpha}}\left( \omega \vert \psi \right) \leq D_\alpha\left( \omega \vert \psi \right) \eqend{,}
\end{equation}
which for $\alpha = 2$ reduces to
\begin{equation}
\label{eq:srel_bound_d2}
S_\mathrm{rel}\left( \omega \vert \psi \right) \leq S_\frac{3}{2}\left( \omega \vert \psi \right) \leq D_2\left( \omega \vert \psi \right) = \ln \norm{ \Psi }^4_{4,\Omega}
\end{equation}
in terms of the non-commutative $L^4$ norm computed with respect to the commutant $\mathfrak{M}'$, which is given by~\eqref{eq:lpnorm}
\begin{equation}
\norm{ \Psi }_{4,\Omega} = \sup_{\xi \in \mathcal{H}\colon \norm{\xi} = 1} \norm{ ( \Delta'_{\xi, \Omega} )^\frac{1}{4} \Psi } = \sup_{\xi \in \mathcal{H}\colon \norm{\xi} = 1} \sqrt{ \left( \Psi, ( \Delta'_{\xi, \Omega} )^\frac{1}{2} \Psi \right) } \eqend{.}
\end{equation}
To compute this expectation value, we note that by Thm.~\ref{thm:modular_representative_change} the relative modular operator $\Delta_{\xi, \Omega}'$ does not depend on the representative vector $\xi$, and so we may take the representative in the positive cone: $\xi^+ \in \mathcal{P}$. Taking $\Psi = b' \Omega$ for some $b' \in \mathfrak{M}'$ with $\norm{ b' \Omega } = 1$, we then obtain
\begin{splitequation}
\left( \Psi, ( \Delta'_{\xi^+, \Omega} )^\frac{1}{2} \Psi \right) &= \left( b' \Omega, J^2 ( \Delta'_{\xi^+, \Omega} )^\frac{1}{2} b' \Omega \right) = \left( b' \Omega, J (b')^* \xi^+ \right) \\
&= \left( (b')^* \xi^+, J b' \Omega \right) = \left( \xi^+, b' J b' \Omega \right) = \left( \xi^+, b' J b' J \Omega \right) \eqend{,}
\end{splitequation}
where we used that $J'_{\xi^+, \Omega} = J$~\cite[Thm.~2.4]{araki1977}, the definition~\eqref{eq:relative_tomita_def} of the relative Tomita operator, that $J$ is antilinear, and that $J \Omega = \Omega$. The supremum is clearly obtained for
\begin{equation}
\xi^+ = \frac{b' J b' J \Omega}{\norm{ b' J b' J \Omega }} \in \mathcal{P} \eqend{,}
\end{equation}
which lies in the positive cone~\cite[Remark~1.2]{haagerup1975}, and it follows that
\begin{equation}
\norm{ b' \Omega }_{4,\Omega} = \norm{ b' J b' J \Omega }^\frac{1}{2} = \norm{ b' J b' \Omega }^\frac{1}{2} \eqend{.}
\end{equation}
Since $b' J b' \Omega \in \mathcal{P}$ such that $b' J b' \Omega = J b' J b' \Omega$, we can further rewrite this as
\begin{splitequation}
\norm{ b' J b' \Omega }^2 &= \left( b' J b' \Omega, b' J b' \Omega \right) = \left( J b' J b' \Omega, b' J b' \Omega \right) \\
&= \left( b' \Omega, J (b')^* J b' J b' \Omega \right) = \left( b' \Omega, b' J (b')^* J J b' \Omega \right) \\
&= \left( b' \Omega, b' J (b')^* b' \Omega \right) = \left( (b')^* b' \Omega, J F (b')^* b' \Omega \right) \\
&= \left( (b')^* b' \Omega, (\Delta')^\frac{1}{2} (b')^* b' \Omega \right) = \norm{ (\Delta')^\frac{1}{4} (b')^* b' \Omega }^2 \eqend{,}
\end{splitequation}
where we used that $J (b')^* J \in \mathcal{M}$, and thus commutes with $b'$. It follows that
\begin{equation}
\label{eq:l4_norm_explicit}
\norm{ b' \Omega }^4_{4,\Omega} = \norm{ b' J b' \Omega }^2 = \norm{ (\Delta')^\frac{1}{4} (b')^* b' \Omega }^2 = \norm{ \Delta^{- \frac{1}{4}} (b')^* b' \Omega }^2 \eqend{,}
\end{equation}
which together with the bound~\eqref{eq:srel_bound_d2} yields the middle inequality in~\eqref{eq:bound1_inequality} of Thm.~\ref{thm:bound1}. We note that~\eqref{eq:l4_norm_explicit} can alternatively be obtained from~\cite[Eq.~(151)]{hollands2023} by taking $a = \1$, $\eta = \psi = \Omega$ and $n = 4$ there, but the present more elementary derivation is sufficient for our needs. In particular, we see that the non-commutative $L^4$ norm can be expressed using the \emph{non-relative} modular operator $\Delta$, even though its definition~\eqref{eq:lpnorm} is in terms of the relative one.

\begin{remark}
In the case of a commutative von Neumann algebra $\mathfrak{A}$ (with faithful normal state $\omega$), the Gelfand--Neumark theorem combined with the Riesz--Markov theorem shows that $\mathfrak{A}$ can be represented as $L^\infty(X,\mu)$ with $X$ a compact Hausdorff space and $\mu$ a probability measure, acting by multiplication on the Hilbert space $L^2(X,\mu)$. The state $\omega$ is implemented by the function $\Omega = 1$, the modular operator $\Delta = 1$ and the modular conjugation $J$ acts by complex conjugation of functions. The relative modular operator is then given by the quotient of the implementing functions in the positive cone $\mathcal{P} = L^2(X,\mu)_+$, which are the positive elements of $L^2(X,\mu)$, namely functions which are positive almost everywhere. A straightforward computation then shows that the non-commutative $L^p$ norms~\eqref{eq:lpnorm} reduce to the usual $L^p$ norms on $(X,\mu)$. In particular, the identity~\eqref{eq:l4_norm_explicit} gives the evident identity $\norm{ f }^4_4 = \norm{ f^* f }_2^2$, i.e., the $L^4$ norm of $f$ coincides with the square root of the $L^2$ norm of the square $f^* f = \abs{ f }^2$ of the element $\abs{ f } \in \mathcal{P}$ of the positive cone canonically associated to $f$.
\end{remark}

Taking now the limit $\alpha \to \infty$ in Lemma~\ref{lem:arakimasuda_increasing}, we obtain the bound
\begin{equation}
D_2\left( \omega \vert \psi \right) \leq \lim_{\alpha \to \infty} D_\alpha\left( \omega \vert \psi \right) = \lim_{\alpha \to \infty} \left( \frac{2 \alpha}{\alpha-1} \ln \norm{ \Psi }_{2 \alpha,\Omega} \right) = 2 \ln \norm{ \Psi }_{\infty,\Omega} \eqend{.}
\end{equation}
For the $L^\infty$ norm of $\Psi = b' \Omega$ we also obtain an explicit expression by computing
\begin{splitequation}
\norm{ b' \Omega }_{\infty,\Omega} &= \sup_{\xi \in \mathcal{H}\colon \norm{\xi} = 1} \norm{ ( \Delta'_{\xi, \Omega} )^\frac{1}{2} b' \Omega } = \sup_{\xi \in \mathcal{H}\colon \norm{\xi} = 1} \norm{ J'_{\xi, \Omega} ( \Delta'_{\xi, \Omega} )^\frac{1}{2} b' \Omega } \\
&= \sup_{\xi \in \mathcal{H}\colon \norm{\xi} = 1} \norm{ (b')^* \xi } = \norm{ (b')^* }_\mathrm{op} = \norm{ b' }_\mathrm{op} \eqend{,}
\end{splitequation}
where the last equality (of the operatorial norms) holds in any von Neumann algebra. Together with the bound~\eqref{eq:srel_bound_d2}, this yields the right inequality in~\eqref{eq:bound1_inequality} of Thm.~\ref{thm:bound1}. If $b'$ is not an element of $\mathfrak{M}'$, but only affiliated with it, the operatorial norm $\norm{ b' }_\mathrm{op}$ diverges, but the $L^4$ norm~\eqref{eq:l4_norm_explicit} is still finite if $\Omega \in \dom{ (b')^* b' }$~\cite[Prop.~B.2]{bostelmanncadamurosangaletti2025}. Therefore, the middle inequality in~\eqref{eq:bound1_inequality} is still non-trivial, and Thm.~\ref{thm:bound1} is proven. In particular, the relative entropy

For a fixed element of the commutant algebra $b' \in \mathfrak{M}'$, we then consider the state $\omega_{b'}$ on $\mathfrak{M}$ defined by
\begin{equation}
\omega_{b'}(a) = \frac{\omega\left( (b')^* a b' \right)}{\omega\left( (b')^* b' \right)} = \frac{\left( b' \Omega, a b' \Omega \right)}{\norm{ b' \Omega }^2} \eqend{,} \quad a \in \mathfrak{M} \eqend{.}
\end{equation}
It follows straightforwardly that $\omega_{b'}$ is normal, and we see that it is controlled by $\omega$ (in symbols $\omega_{b'} \lll \omega$), i.e.,
\begin{equation}
\omega_{b'}(a) \leq C \omega(a) \quad\forall a \in \mathfrak{M}_+
\end{equation}
for some positive constant $C$ independent of $a$. Indeed, if $a = c^* c$ is a positive element of $\mathfrak{M}$, we have
\begin{equation}
\omega_{b'}\left( c^* c \right) = \frac{\norm{ c b' \Omega }^2}{\norm{ b' \Omega }^2} = \frac{\norm{ b' c \, \Omega }^2}{\norm{ b' \Omega }^2} \leq \frac{\norm{ b' }_\mathrm{op}^2}{\norm{ b' \Omega }^2} \norm{ c \, \Omega }^2 = \frac{\norm{ b' }_\mathrm{op}^2}{\norm{ b' \Omega }^2} \, \omega\left( c^* c \right) \eqend{,}
\end{equation}
such that $\omega_{b'} \lll \omega$ with the constant $C = \norm{ b' }_\mathrm{op}^2 \norm{ b' \Omega }^{-2}$. Actually, the opposite implication is also always true, namely for every normal state $\chi \lll \omega$ it exists a $b' \in \mathfrak{M}'$ such that $\chi$ is implemented by the vector $b' \Omega$ in the GNS representation induced by $\omega$~\cite[Thm.~2.3.19]{brattelirobinson1}. This is the condition considered by Berta, Scholz and Tomamichel, who show that~\cite[Lem.~8 and Thm.~13]{bertascholztomamichel2018}
\begin{equation}
D^\mathrm{BST}(\rho \Vert \sigma) \leq D^\mathrm{BST}_\alpha(\rho \Vert \sigma) = \frac{1}{\alpha-1} \ln \norm{\rho}_{2 \alpha, \sigma}^{2 \alpha}
\end{equation}
for all $\alpha > 1$ if $\rho \lll \sigma$. Taking $\sigma = \omega$ and $\rho = \psi = \omega_{b'}$, and taking into account that $D^\mathrm{BST}(\rho \Vert \sigma) = S_\mathrm{rel}(\sigma \vert \rho)$ and $D^\mathrm{BST}_\alpha(\rho \Vert \sigma) = D_\alpha(\sigma \vert \rho)$, this is exactly the bound~\eqref{eq:srel_bound_dalpha}. We note that Jen{\v c}ov{\'a}~\cite{jencova2018} has already extended this inequality to states which do not satisfy $\rho \lll \sigma$; our contributions are the explicit expressions for the $L^4$ and the $L^\infty$ norm involving only the non-relative modular operator of the state $\omega$. Moreover, we have shown that the bound in terms of the $L^4$ norm is non-trivial even for unbounded operators affiliated with $\mathfrak{M}'$ under mild domain conditions. In particular, these conditions are satisfied by polynomials of the fields in a Wightman QFT for the Minkowski vacuum state. Namely, the fields are required to share a common dense and invariant domain that contains $\Omega$, from which it also follows that polynomials of the fields are closable operators and, by causality, that they are affiliated with local von Neumann algebras~\cite[Sec.~2.4]{sangalettithesis}.

We finally note that if $b' = u' \in \mathfrak{M}'$ is unitary, the inequality~\eqref{eq:bound1_inequality} reduces to
\begin{equation}
\label{eq:bound1_inequality_uprime}
0 \leq S_\mathrm{rel}\left( \Omega \Vert u' \Omega \right) \leq 2 \ln \norm{ \Delta^{- \frac{1}{4}} \Omega } = 2 \ln \norm{ \Omega } = 0 \eqend{,}
\end{equation}
such that the relative entropy vanishes. This is consistent with the fact that $\Omega$ and $u' \Omega$ implement the same state $\omega$ on $\mathfrak{M}$:
\begin{equation}
\omega_{u'}(a) = \left( u' \Omega, a u' \Omega \right) = \left( \Omega, a (u')^* u' \Omega \right) = \left( \Omega, a \Omega \right) = \omega(a) \quad\forall a \in \mathfrak{M} \eqend{.}
\end{equation}

\subsection{Analytic elements}

To obtain from Thm.~\ref{thm:bound1} a bound also for excitations from the algebra, we need to consider elements $a \in \mathfrak{M}$ which are analytic with respect to the modular flow, i.e., for which the function $t \mapsto \sigma_t(a) = \Delta^{\mathi t} a \Delta^{- \mathi t}$ for $t \in \mathbb{R}$ extends to an analytic function in $\mathbb{C}$. We denote the $*$-subalgebra of $\mathfrak{M}$ consisting of analytic elements by $\mathfrak{M}_\sigma$. For $a \in \mathfrak{M}_\sigma$, we obtain
\begin{equation}
a \Omega = J \Delta^\frac{1}{2} a^* \Omega = J \Delta^\frac{1}{2} a^* \Delta^{- \frac{1}{2}} J \Omega = J \left[ \Delta^{- \frac{1}{2}} a \Delta^\frac{1}{2} \right]^* J \Omega = J \left[ \sigma_\frac{\mathi}{2}(a) \right]^* J \Omega \eqend{,}
\end{equation}
and since $a$ is analytic with respect to the modular flow and $\sigma_t(a) \in \mathfrak{M}$ for all real $t$, we have $\sigma_ \frac{\mathi}{2}(a) \in \mathfrak{M}$ and $b' = J \left[ \sigma_\frac{\mathi}{2}(a) \right]^* J \in \mathfrak{M}'$, which thus furnishes the required element from the commutant. Note that the equality $\sigma_\frac{\mathi}{2}(a) = \Delta^{- \frac{1}{2}} a \Delta^\frac{1}{2}$ holds only in a weak sense, namely as an identity for quadratic forms on the set of analytic vectors for $\Delta$~\cite[Sec.~$2.5.3.$]{brattelirobinson1}. Nevertheless, the vectors $a \Omega$ and $b' \Omega$ are identical, and in particular give the same state (assuming that $\norm{ a \Omega } = 1$, which can always be achieved by rescaling $a$). Therefore, we have
\begin{equation}
S_\mathrm{rel}\left( \Omega \Vert a \Omega \right) = S_\mathrm{rel}\left( \Omega \Vert b' \Omega \right) \eqend{,}
\end{equation}
and can apply the bounds of Thm.~\ref{thm:bound1} to the latter one. It follows that
\begin{equation}
\label{eq:srel_analytic_bound}
0 \leq S_\mathrm{rel}\left( \Omega \Vert a \Omega \right) \leq 2 \ln \norm{ \Delta^{- \frac{1}{4}} J \left[ \sigma_\frac{\mathi}{2}(a) \right] J J \left[ \sigma_\frac{\mathi}{2}(a) \right]^* J \Omega } \leq 2 \ln \norm{ J \left[ \sigma_\frac{\mathi}{2}(a) \right]^* J }_\mathrm{op} \eqend{,}
\end{equation}
and the middle norm can be rewritten as
\begin{splitequation}
\label{eq:srel_analytic_bound_1}
\norm{ \Delta^{- \frac{1}{4}} J \left[ \sigma_\frac{\mathi}{2}(a) \right] J J \left[ \sigma_\frac{\mathi}{2}(a) \right]^* J \Omega } &= \norm{ \Delta^{- \frac{1}{4}} J \left[ \sigma_\frac{\mathi}{2}(a) \right] \left[ \sigma_\frac{\mathi}{2}(a) \right]^* \Omega } \\
&= \norm{ J \Delta^{- \frac{1}{4}} J \Delta^{- \frac{1}{2}} a \Delta a^* \Omega } = \norm{ \Delta^{- \frac{1}{4}} a \Delta a^* \Omega } \\
&= \norm{ \Delta^{- \frac{1}{4}} a \Delta^\frac{1}{4} \left[ \Delta^{- \frac{3}{4}} a \Delta^\frac{3}{4} \right]^* \Omega } = \norm{ \sigma_\frac{\mathi}{4}(a) \left[ \sigma_\frac{3 \mathi}{4}(a) \right]^* \Omega } \eqend{,}
\end{splitequation}
where we used that $J \Delta^r J = \Delta^{-r}$. For the operational norm we compute
\begin{splitequation}
\label{eq:srel_analytic_bound_2}
\norm{ J \left[ \sigma_\frac{\mathi}{2}(a) \right]^* J }_\mathrm{op} &= \sup_{\xi \in \mathcal{H} \colon \norm{\xi} = 1} \norm{ J \left[ \sigma_\frac{\mathi}{2}(a) \right]^* J \xi } = \sup_{\hat{\xi} \in \mathcal{H} \colon \norm{\hat{\xi}} = 1} \norm{ \left[ \sigma_\frac{\mathi}{2}(a) \right]^* \hat{\xi} } \\
&= \norm{ \left[ \sigma_\frac{\mathi}{2}(a) \right]^* }_\mathrm{op} = \norm{ \sigma_\frac{\mathi}{2}(a) }_\mathrm{op} \eqend{,}
\end{splitequation}
where we used that $J$ is invertible and $J^2 = \1$, which then proves Cor.~\ref{corr:bound2}.

While being analytic seems to be a strong restriction on $a$, the set of analytic elements $\mathfrak{M}_\sigma$ is actually dense in $\mathfrak{M}$ in the $\sigma$-weak topology~\cite[Cor.~2.5.23]{brattelirobinson1}. In turn, this implies that the set of vectors $a \Omega$ with $a \in \mathfrak{M}_\sigma$ is dense in the Hilbert space $\mathcal{H}$. Concretely, we have
\begin{lemma}[{\cite[Lemma~3.2.5]{sangalettithesis}}]
\label{lem:density_analytic}
The vector $\Omega$ (cyclic and separating for $\mathfrak{M}$) is cyclic and separating for the $*$-subalgebra $\mathfrak{M}_\sigma \subset \mathfrak{M}$ of analytic elements of $\mathfrak{M}$.
\end{lemma}
\begin{proof}
We recall that a vector is cyclic for a unital $*$-algebra if and only if it is separating for its commutant. Since $\mathfrak{M}_\sigma \subset \mathfrak{M}$ and $\Omega$ is separating for $\mathfrak{M}$, it is separating also for $\mathfrak{M}_\sigma$. From the bicommutant theorem~\cite[Cor.~2.4.15]{brattelirobinson1}, it follows that
\begin{equation}
\mathfrak{M}_\sigma' = \left( \mathfrak{M}_\sigma'' \right)' = \mathfrak{M}' \eqend{,}
\end{equation}
using that $\mathfrak{M}_\sigma$ is $\sigma$-weakly dense in $\mathfrak{M}$. Since $\Omega$ is separating for $\mathfrak{M}'$, it is therefore separating for $\mathfrak{M}_\sigma'$, and thus cyclic for $\mathfrak{M}_\sigma$.
\end{proof}
For any $a \in \mathfrak{M}$, we may in fact explicitly construct an approximating sequence $(a_n)_{n \in \mathbb{N}}$ of analytic elements $a_n \in \mathfrak{M}_\sigma$ by smearing over the modular flow of $a$. Namely, we set
\begin{equation}
\label{eq:analytic_sequence_def}
a_n \coloneq \sqrt{ \frac{n}{\pi} } \int_{-\infty}^\infty \mathe^{- n s^2} \sigma_s(a) \total s \eqend{,}
\end{equation}
and by~\cite[Prop. 2.5.22]{brattelirobinson1} have $a_n \Omega \to a \Omega$ in norm as $n \to \infty$. For these elements, we have
\begin{equation}
\sigma_t(a_n) = \sqrt{ \frac{n}{\pi} } \int_{-\infty}^\infty \mathe^{- n s^2} \sigma_{s+t}(a) \total s = \sqrt{ \frac{n}{\pi} } \int_{-\infty}^\infty \mathe^{- n (s-t)^2} \sigma_s(a) \total s
\end{equation}
for all $t \in \mathbb{R}$, and in the latter integral can analytically continue the integrand to arbitrary $t \in \mathbb{C}$. In particular, we obtain
\begin{equation}
\label{eq:analytic_flow}
\sigma_{\mathi r}(a_n) = \sqrt{ \frac{n}{\pi} } \int_{-\infty}^\infty \mathe^{- n \left( s - \mathi r \right)^2} \sigma_s(a) \total s
\end{equation}
and
\begin{splitequation}
\label{eq:analytic_flow_norm}
\norm{ \sigma_\frac{\mathi}{4}(a_n) \left[ \sigma_\frac{3 \mathi}{4}(a_n) \right]^* \Omega }^2 &= \left( \frac{n}{\pi} \right)^2 \int\dotsi\int_{-\infty}^\infty \mathe^{ - n \left( s_1 + \frac{\mathi}{4} \right)^2 - n \left( s_2 - \frac{3 \mathi}{4} \right)^2 - n \left( s_3 - \frac{\mathi}{4} \right)^2 - n \left( s_4 + \frac{3 \mathi}{4} \right)^2} \\
&\qquad\qquad\times \Bigl( \sigma_{s_1}(a) \sigma_{s_2}(a^*) \Omega, \sigma_{s_3}(a) \sigma_{s_4}(a^*) \Omega \Bigr) \total s_1 \total s_2 \total s_3 \total s_4 \eqend{.}
\end{splitequation}

However, note that even if we can approximate vectors arbitrarily well, we generally cannot do that for the relative entropy because of the unboundedness of $\Delta$. In addition, note that for unitary $a \in \mathfrak{M}$ with a non-vanishing relative entropy $S_\mathrm{rel}(\Omega \vert a \Omega) \neq 0$, the corresponding $b' \in \mathfrak{M}'$ cannot be unitary, since otherwise the vector state $b' \Omega$ on $\mathfrak{M}$ would coincide with $\Omega$, contradicting the fact that $S_\mathrm{rel}(\Omega \vert a \Omega) \neq 0$. Nevertheless, since the relative entropy is lower semicontinuous~\cite[Thm.~(3.7)]{araki1977}, we have
\begin{equation}
S_\mathrm{rel}\left( \Omega \Vert a \Omega \right) \leq \liminf_{n \to \infty} S_\mathrm{rel}\left( \Omega \Vert a_n \Omega \right) \eqend{.}
\end{equation}
Therefore, the bounds
\begin{equation}
0 \leq S_\mathrm{rel}\left( \Omega \Vert a_n \Omega \right) \leq 2 \ln \norm{ \sigma_\frac{\mathi}{4}(a_n) \left[ \sigma_\frac{3 \mathi}{4}(a_n) \right]^* \Omega } \leq 2 \ln \norm{ \sigma_\frac{\mathi}{2}(a_n) }_\mathrm{op}
\end{equation}
obtained from Cor.~\ref{corr:bound2} (or equations~\eqref{eq:srel_analytic_bound}--\eqref{eq:srel_analytic_bound_2}) for $a_n$ also give a bound on the relative entropy between $\Omega$ and $a \Omega$, namely
\begin{equation}
S_\mathrm{rel}\left( \Omega \Vert a \Omega \right) \leq 2 \liminf_{n \to \infty} \ln \norm{ \sigma_\frac{\mathi}{4}(a_n) \left[ \sigma_\frac{3 \mathi}{4}(a_n) \right]^* \Omega } \eqend{,}
\end{equation}
which then proves Cor.~\ref{corr:bound3}. In general, it is not guaranteed that the bound is nontrivial since the right-hand side might diverge as $n \to \infty$, but we will see that it actually is finite in an example.

\section{Examples}
\label{sec:example}

Lastly, we want to show that the bound of Cor.~\ref{corr:bound3} gives a non-trivial result in a specific example. For this, we consider the free scalar field in Minkowski spacetime in a wedge, and the chiral current on a light ray. For both examples, we determine analytic vectors and their swapping partners, but only derive a finite bound in the latter case, where we can evaluate the required integrals almost explicitly.

\subsection{Free scalar field in a wedge}

Consider a free scalar field in Minkowski spacetime in the standard right wedge $\mathcal{W}_\mathrm{r} \coloneq \{ x \in \mathbb{R}^{d+1} \colon x^1 \geq \abs{ x^0 } \}$. The Bisognano--Wichmann theorem~\cite{bisognanowichmann1975,bisognanowichmann1976} shows that the vacuum vector $\Omega$ in the GNS representation of the vacuum state $\omega$ is cyclic and separating for the von Neumann algebra $\mathfrak{M}$ generated by Weyl operators $W(f)$ with real Schwartz functions $f \in \mathcal{S}_\mathbb{R}(\mathcal{W}_\mathrm{r})$ localized in the right wedge, and the modular operator for the pair $(\mathfrak{M}, \Omega)$ coincides with the boost operator in the positive $x^1$ direction. Moreover, the commutant $\mathfrak{M}'$ is the von Neumann algebra generated by Weyl operators $W(f)$ with real Schwartz functions $f \in \mathcal{S}_\mathbb{R}(\mathcal{W}_\mathrm{l})$ localized in the left wedge $\mathcal{W}_\mathrm{l} \coloneq \{ x \in \mathbb{R}^{d+1} \colon x^1 \leq - \abs{ x^0 } \}$. Since the theory is free, the modular objects are second-quantized operators on Fock space~\cite{figlioliniguido1989,figlioliniguido1994}, such that
\begin{equation}
\label{eq:boost_modular_action}
\Delta_\Omega^{\mathi t} W(f) \Delta_\Omega^{-\mathi t} = W(f_t) \quad\text{with}\quad f_t(x) \coloneq f(\Lambda_{-t} x) \ \forall t \in \mathbb{R} \eqend{,}
\end{equation}
where we denote with $\Lambda_{-t}$ the boosted coordinates in the $x^1$ direction:
\begin{equations}[eq:boosted_coords]
(\Lambda_t x)^0 &= \cosh(2 \pi t) x^0 + \sinh(2 \pi t) x^1 \eqend{,} \\
(\Lambda_t x)^1 &= \sinh(2 \pi t) x^0 + \cosh(2 \pi t) x^1 \eqend{,} \\
(\Lambda_t x)^i &= x^i \quad\text{for}\quad i \in \{ 2,\ldots,d+1 \} \eqend{.}
\end{equations}
Since for $x \in \mathcal{W}_\mathrm{r}$ we have
\begin{equation}
\label{eq:boost_wedge_invariant}
(\Lambda_t x)^1 \pm (\Lambda_t x)^0 = \mathe^{\pm 2 \pi t} \left( x^1 \pm x^0 \right) \geq 0 \eqend{,}
\end{equation}
the boosts leave the right wedge $\mathcal{W}_\mathrm{r}$ invariant (and analogously for the left wedge), and we have an instance of a local geometric action of the modular group~\cite{buchholzsummers1993,buchholzdreyerflorigsummers2000}. Moreover, the modular conjugation $J$ acts on Weyl operators according to
\begin{equation}
\label{eq:wedge_modular_conjugation}
J W(f) J = W(-f_J) \quad\text{with}\quad f_J(x^0, x^1, \ldots, x^{d+1}) \coloneq f(-x^0, -x^1, x^2, \ldots, x^{d+1}) \eqend{,}
\end{equation}
and is extended by complex antilinearity.

The integral kernel of the two-point function of the state $\omega$ is given by
\begin{splitequation}
\label{eq:omega_2pf}
\omega_2(x, y) &= 2 \pi \int \Theta(p^0) \delta(p^2 + m^2) \, \mathe^{- \mathi p (x-y)} \frac{\total^{d+1} p}{(2 \pi)^{d+1}} \\
&= \lim_{\epsilon \to 0^+} \int \frac{1}{2 \omega_\vec{p}} \mathe^{\mathi \omega_\vec{p} (x^0-y^0 + \mathi \epsilon) - \mathi \vec{p} (\vec{x}-\vec{y})} \frac{\total^d \vec{p}}{(2 \pi)^d} \eqend{,}
\end{splitequation}
where $p^2 = - (p^0)^2 + \vec{p}^2$ and $\omega_\vec{p} = \sqrt{ \vec{p}^2 + m^2 }$, and the integral kernel of the commutator function $E$ reads
\begin{splitequation}
\label{eq:commutator_def}
E(x, y) &\coloneq - \mathi \Bigl[ \omega_2(x, y) - \omega_2(y, x) \Bigr] = - E(y, x) \\
&= \lim_{\epsilon \to 0^+} \int \frac{\sin\left[ \omega_\vec{p} (x^0-y^0) - \vec{p} (\vec{x}-\vec{y}) \right]}{\omega_\vec{p}} \mathe^{- \epsilon \, \omega_\vec{p}} \frac{\total^d \vec{p}}{(2 \pi)^d} \eqend{.}
\end{splitequation}
The Weyl relations
\begin{equation}
\label{eq:weyl_relation}
W(f) W(g) = \mathe^{- \frac{\mathi}{2} E(f, g)} W(f+g) \eqend{,} \quad [ W(f) ]^* = W(-f) \eqend{,} \quad
\end{equation}
show that $\{ W(s f) \}_{s \in \mathbb{R}}$ for fixed $f$ form a one-parameter unitary group, whose generator is the field operator $\phi(f)$, which is essentially self-adjoint on the dense domain of finite particle vectors~\cite[Prop.~5.2.3]{brattelirobinson2}. Since finite-particle vectors are analytic for $\phi(f)$, we have
\begin{equation}
\label{eq:weyl_expansion}
W(f) \Omega = \lim_{N \to \infty} \sum_{n=0}^N \frac{\mathi^n}{n!} [ \phi(f) ]^n \Omega = \mathe^{- \norm{ f }^2} \lim_{N \to \infty} \sum_{n=0}^N \frac{1}{n!} \left[ a^*(f) \right]^n \Omega \eqend{,}
\end{equation}
where $a^*(f)$ are the smeared creation operators in terms of which $\phi(f) = \overline{a(f) + a^*(f)}$~\cite[Prop.~5.2.4]{brattelirobinson2} and we used the canonical commutation relations in the second step.

For coherent excitations of the vacuum of the form $W(f) \Omega$, the relative entropy $S_\mathrm{rel}\left( \Omega \Vert W(f) \Omega \right)$ can be explicitly computed~\cite{casinigrillopontello2019,longo2019,ciollilongoruzzi2020}, as well as for other coherent or thermal states (see for example Refs.~\cite{dragofaldinopinamonti2018,hollands2020,panebianco2020,ariasetal2020,bostelmanncadamurodelvecchio2022,guimaraesroditisorellavieira2025,cadamurofroebkatsinismandrysch2025} and references therein). However, taking states of the form $[ \phi(f) ]^k \Omega$ or other non-unitary excitations of the vacuum, this is not the case anymore and only perturbative expansions are available, see for example Refs.~\cite{araki1973,derezinskijaksicpillet2003,lashkariliurajagopal2021,dasezhuthachan2019,lashkariliurajagopal2023,kabatlifschytznguyensarkar2020,jiangkabatmarthandansarkas2025}. Nevertheless, we can bound the relative entropy between such states and the vacuum using Cor.~\ref{corr:bound2}. To do so, consider for a function $f \in L^1_\mathbb{R}(\mathcal{W}_\mathrm{r})$ and $n \in \mathbb{N}$ the function $f_n$ defined by
\begin{equation}
\label{eq:analytic_fn_def}
f_n(x) \coloneq \sqrt{ \frac{n}{\pi} } \int_{-\infty}^\infty \mathe^{- n s^2} f(\Lambda_s x) \total s \eqend{,}
\end{equation}
and its \emph{swapping partner}
\begin{equation}
\label{eq:analytic_fn_swapping}
\left( J \Lambda_\frac{\mathi}{2} f_n \right)(x) \coloneq \sqrt{ \frac{n}{\pi} } \int_{-\infty}^\infty \mathe^{- n \left( s - \frac{\mathi}{2} \right)^2} f\left( - (\Lambda_s x)^0 , - (\Lambda_s x)^1, x^2, \ldots, x^{d+1} \right) \total s \eqend{,}
\end{equation}
whose definition is motivated by the construction~\eqref{eq:analytic_sequence_def} of analytic elements and the action~\eqref{eq:wedge_modular_conjugation} of the modular conjugation $J$. Note that since $\supp f \subset \mathcal{W}_\mathrm{r}$ and the Lorentz boosts leave the wedge invariant~\eqref{eq:boost_wedge_invariant}, the swapping partner $J \Lambda_\frac{\mathi}{2} f_n$ is supported in the left wedge $\mathcal{W}_\mathrm{l}$. These two functions define the same one-particle state $a^*(f_n) \Omega = a^*\left( J \Lambda_\frac{\mathi}{2} f_n \right) \Omega$ if the Fourier transforms of $f_n$ and $J \Lambda_\frac{\mathi}{2} f_n$ agree on the positive mass shell $p^0 = \omega_\vec{p}$. We compute
\begin{splitequation}
\label{eq:fourier_fn}
( \mathcal{F} f_n )(\omega_\vec{p}, \vec{p}) &= \int f_n(x) \, \mathe^{\mathi p x} \total^{d+1} x = \sqrt{ \frac{n}{\pi} } \int_{-\infty}^\infty \int \mathe^{- n s^2} f(\Lambda_s x) \, \mathe^{\mathi p x} \total^{d+1} x \total s \\
&= \sqrt{ \frac{n}{\pi} } \int_{-\infty}^\infty \mathe^{- n s^2} \int f(x) \, \mathe^{\mathi p \Lambda_{-s} x} \total^{d+1} x \total s \\
&= \sqrt{ \frac{n}{\pi} } \int_{-\infty}^\infty \mathe^{- n s^2} \int f(x) \exp\left[ - \mathi H(x, -s, \vec{p}) + \mathi \sum_{j=2}^{d+1} p^j x^j \right] \total^{d+1} x \total s
\end{splitequation}
with the function
\begin{splitequation}
\label{eq:h_function_def}
H(x, z, \vec{p}) &\coloneq \omega_\vec{p} \left[ \cosh(2 \pi z) x^0 + \sinh(2 \pi z) x^1 \right] - p^1 \left[ \sinh(2 \pi z) x^0 + \cosh(2 \pi z) x^1 \right] \\
&= \cosh(2 \pi z) \left( \omega_\vec{p} x^0 - p^1 x^1 \right) + \sinh(2 \pi z) \left( \omega_\vec{p} x^1 - p^1 x^0 \right) \eqend{,}
\end{splitequation}
where we could interchange the integrals using Fubini's theorem since $f \in L^1_\mathbb{R}(\mathcal{W}_\mathrm{r})$, and then used that the Jacobi determinant of the change $x \to \Lambda_{-s} x$ is 1. For complex $z$, we compute using well-known hyperbolic function identities
\begin{splitequation}
\Im H(x, z, \vec{p}) &= \sin\left( 2 \pi \Im z \right) \biggl[ \left( \omega_\vec{p} - \abs{p^1} \right) \left[ \cosh\left( 2 \pi \Re z \right) \left( x^1 - \abs{x^0} \right) + \mathe^{2 \pi \sgn x^0 \Re z} \abs{x^0} \right] \\
&\hspace{7em}+ \abs{p^1} \mathe^{- 2 \pi \sgn p^1 \Re z} \left( x^1 - \sgn p^1 x^0 \right) \biggr] \eqend{,}
\end{splitequation}
and since for $x \in \mathcal{W}_\mathrm{r}$ we have $x^1 \geq \pm x^0$ all the terms in brackets are non-negative. It follows that for $\abs{\Im z} \leq \frac{1}{2}$, the sign of $\Im H(x, z, p)$ is the same as the sign of $\Im z$, and hence the integrand in~\eqref{eq:fourier_fn} can be analytically continued to $\Im s \in \left[ - \frac{1}{2}, 0 \right]$ where $\mathe^{- \mathi H(x, -s, \vec{p})}$ is bounded. Using that
\begin{equation}
H\left( x, -s-\frac{\mathi}{2}, \vec{p} \right) = - H(x, -s, \vec{p}) = - \omega_\vec{p} ( \Lambda_{-s} x )^0 + p^1 ( \Lambda_{-s} x )^1 \eqend{,}
\end{equation}
we thus obtain
\begin{splitequation}
\label{eq:fourier_fn_final}
( \mathcal{F} f_n )(\omega_\vec{p}, \vec{p}) &= \sqrt{ \frac{n}{\pi} } \int_{-\infty}^\infty \mathe^{- n \left( s - \frac{\mathi}{2} \right)^2} \int f(x) \exp\left[ - \mathi H\left( x, s-\frac{\mathi}{2}, \vec{p} \right) + \mathi \sum_{j=2}^{d+1} p^j x^j \right] \total^{d+1} x \total s \\
&= \sqrt{ \frac{n}{\pi} } \int_{-\infty}^\infty \mathe^{- n \left( s - \frac{\mathi}{2} \right)^2} \int f(x) \, \mathe^{\mathi \omega_\vec{p} ( \Lambda_{-s} x )^0 - \mathi p^1 ( \Lambda_{-s} x )^1 + \mathi \sum_{j=2}^{d+1} p^j x^j} \total^{d+1} x \total s \eqend{.}
\end{splitequation}

On the other hand, for the swapping partner $J \Lambda_\frac{\mathi}{2} f_n$ we compute
\begin{splitequation}
\left( \mathcal{F} J \Lambda_\frac{\mathi}{2} f_n \right)(\omega_\vec{p}, \vec{p}) &= \int \left( J \Lambda_\frac{\mathi}{2} f_n \right)(x) \, \mathe^{\mathi p x} \total^{d+1} x \\
&= \sqrt{ \frac{n}{\pi} } \int_{-\infty}^\infty \mathe^{- n \left( s - \frac{\mathi}{2} \right)^2} \int f\left( - x^0 , - x^1, x^2, \ldots, x^{d+1} \right) \, \mathe^{\mathi p \Lambda_{-s} x} \total^{d+1} x \total s \eqend{,}
\end{splitequation}
following the same steps as for $f_n$. Changing finally $x^0 \to - x^0$ and $x^1 \to - x^1$, this is seen to be equal to~\eqref{eq:fourier_fn_final}, and thus the Fourier transforms of $f_n$ and its swapping partner $J \Lambda_\frac{\mathi}{2} f_n$ agree on the positive mass shell $p^0 = \omega_\vec{p}$. Therefore, even though $f_n$ and $J \Lambda_\frac{\mathi}{2} f_n$ are supported on opposite wedges, such that the fields $\phi(f_n)$ and $\phi\left( J \Lambda_\frac{\mathi}{2} f_n \right)$ are affiliated with the algebra $\mathfrak{M}$ and its commutant $\mathfrak{M}'$, respectively, the states $\phi(f_n) \Omega = a^*(f_n) \Omega$ and $\phi\left( J \Lambda_\frac{\mathi}{2} f_n \right) \Omega = a^*\left( J \Lambda_\frac{\mathi}{2} f_n \right) \Omega$ are the same. We remark that swapping partners for unbounded operators have previously been constructed in the literature, see for example~\cite[Thm.~3.6]{borchers1995}, \cite[Eq.~(3.13)]{lechner2003} or~\cite[Lemma~3]{duell2017}, and the construction presented here is essentially the same; for coherent states see also~\cite{falconeconti2025}.

Note that while $f_n$ is real-valued (and $\phi(f_n)$ is self-adjoint), its swapping partner is not, and thus their Fourier transforms on the negative mass shell $p^0 = - \omega_\vec{p}$ are different. This is in fact needed to avoid a contradiction with the fact that $\mathfrak{M}$ is a factor, since if the Fourier transforms would agree on both mass hyperboloids we would have $\phi(f_n) = \phi\left( J \Lambda_\frac{\mathi}{2} f_n \right)$ as operators. Since the first is affiliated with $\mathfrak{M}$ and the second with the commutant $\mathfrak{M}'$, this would only be possible if they are equal to the identity $\1$. In general, indeed, if $a$ is an analytic self-adjoint operator, its swapping partner $J \Delta^\frac{1}{2} a \Delta^{-\frac{1}{2}} J$ is not necessarily self-adjoint, since its (formal) adjoint reads $J \Delta^{-\frac{1}{2}} a \Delta^\frac{1}{2} J$, where the exponents are changed. In fact, if $\phi(f_n) = \phi\left( J \Lambda_\frac{\mathi}{2} f_n \right)$ holds, also the Weyl operators $W(f_n)$ and $W\left( J \Lambda_\frac{\mathi}{2} f_n \right)$ would be the same, and since the first is an element of $\mathfrak{M}$ and the second of its commutant, they would need to be equal to the identity since $\mathfrak{M}$ is a factor. Nevertheless, since finite-particle vectors are analytic for $\phi(f)$, and the creation operator $a^*(f)$ only depends on $f$ restricted to the positive mass shell, we obtain from equation~\eqref{eq:weyl_expansion} that
\begin{splitequation}
W\left( J \Lambda_\frac{\mathi}{2} f_n \right) \Omega &= \mathe^{- \norm{ J \Lambda_\frac{\mathi}{2} f_n }^2} \lim_{N \to \infty} \sum_{k=0}^N \frac{1}{k!} \left[ a^*\left( J \Lambda_\frac{\mathi}{2} f_n \right) \right]^k \Omega \\
&= \mathe^{- \norm{ J \Lambda_\frac{\mathi}{2} f_n }^2} \lim_{N \to \infty} \sum_{k=0}^N \frac{1}{k!} \left[ a^*(f_n) \right]^k \Omega = \mathe^{\norm{ f_n }^2 - \norm{ J \Lambda_\frac{\mathi}{2} f_n }^2} W(f_n) \Omega \eqend{,}
\end{splitequation}
i.e., the vectors $W\left( J \Lambda_\frac{\mathi}{2} f_n \right) \Omega$ and $W(f_n) \Omega$ only differ by a rescaling. We thus also obtain the swapping partner of the unitary excitation $W(f_n) \Omega$, which is given by
\begin{equation}
\label{eq:weyl_swapping_partner}
\mathe^{\norm{ J \Lambda_\frac{\mathi}{2} f_n }^2 - \norm{ f_n }^2} W\left( J \Lambda_\frac{\mathi}{2} f_n \right) \Omega \eqend{.}
\end{equation}
In particular, it is \emph{not} a unitary excitation, which is consistent with the fact that the relative entropy $S_\mathrm{rel}\left( \Omega \Vert W(f_n) \Omega \right)$ does not vanish~\cite{casinigrillopontello2019,longo2019,ciollilongoruzzi2020}.

\subsection{Unitary swapping partners}

In general, it is possible to find unitary operators in a von Neumann algebra that admit a unitary swapping partner. For example, consider the Hilbert space $\mathcal{H} = \mathbb{C}^2 \otimes \mathbb{C}^2$ with the usual scalar product, the von Neumann algebra $\mathfrak{A} \coloneq \mathcal{B}(\mathbb{C}^2) \otimes \1$ and its commutant $\mathfrak{A}' = \1 \otimes \mathcal{B}(\mathbb{C}^2)$, where $\mathcal{B}(\mathbb{C}^2)$ denotes the set of all bounded operators on $\mathbb{C}^2$. The (maximally entangled) vector
\begin{equation}
\Omega \coloneq \frac{1}{\sqrt{2}} \Bigl[ (1,0) \otimes (1,0) + (0,1) \otimes (0,1) \Bigr]
\end{equation}
is cyclic and separating for $\mathfrak{A}$. Indeed, let $A \otimes \1$ be a generic operator in $\mathfrak{A}$. Then we compute
\begin{equation}
\left( A \otimes \1 \right) \Omega = \frac{1}{\sqrt{2}} \Bigl[ A (1,0) \otimes (1,0) + A (0,1) \otimes (0,1) \Bigr] \eqend{,}
\end{equation}
and since $(1,0)$ and $(0,1)$ are orthogonal, we obtain the implication
\begin{equation}
\left( A \otimes \1 \right) \Omega = 0 \implies A (1,0) = 0 \ \text{and}\ A (0,1) = 0 \implies A = 0 \eqend{,}
\end{equation}
and $\Omega$ is separating. Similarly, $\Omega$ is separating for $\mathfrak{A}'$ and therefore cyclic for $\mathfrak{A}$. Consider now the operators $U \in \mathfrak{A}$ and $U' \in \mathfrak{A}'$ defined as
\begin{equation}
U \coloneq \sigma_z \otimes \1 \eqend{,} \quad U' \coloneq \1 \otimes \sigma_z \eqend{,} \quad \sigma_z = \begin{pmatrix} 1 & 0 \\ 0 & - 1 \end{pmatrix} \eqend{.}
\end{equation}
As tensor products of unitary operators, $U$ and $U'$ are unitary, and we obtain
\begin{equation}
U \Omega = \frac{1}{\sqrt{2}} \Bigl[ (1,0) \otimes (1,0) - (0,1) \otimes (0,1) \Bigr] = U' \Omega \eqend{.}
\end{equation}
One easily verifies that $U \Omega = U' \Omega$ and $\Omega$ implement the same state on both the von Neumann algebra $\mathfrak{A}$ and the commutant $\mathfrak{A}'$ \emph{separately}, namely we obtain
\begin{equation}
\left( U \Omega, \left( A \otimes \1 \right) U \Omega \right)_\mathcal{H} = \frac{1}{2} \left( (1,0) , A (1,0) \right)_{\mathbb{C}^2} + \frac{1}{2} \left( (0,1) , A (0,1) \right)_{\mathbb{C}^2} = \left( \Omega, \left( A \otimes \1 \right) \Omega \right)_\mathcal{H}
\end{equation}
and analogously $\left( U \Omega, \left( \1 \otimes A \right) U \Omega \right)_\mathcal{H} = \left( \Omega, \left( \1 \otimes A \right) \Omega \right)_\mathcal{H}$.

We thus have succeeded in constructing two unitary operators $U \in \mathfrak{A}$ and $U' \in \mathfrak{A}'$ that implement the same state, hence are swapping partners. Since they implement the same state, the relative entropy between $\Omega$ and $U \Omega$ vanishes, which is in accordance with our bounds of Thm.~\ref{thm:bound1} in terms of the non-commutative $L^p$-norms. Namely, using equation~\eqref{eq:bound1_inequality_uprime} which is obtained from the bound of Thm.~\ref{thm:bound1}, we obtain
\begin{equation}
0 \leq S_\mathrm{rel}\left( \Omega \Vert U' \Omega \right) \leq 2 \ln \norm{ \Delta^{- \frac{1}{4}} \Omega } = 2 \ln \norm{ \Omega } = 0 \eqend{,}
\end{equation}
which forces the relative entropy to vanish. It is easy to see that this follows straightforwardly from the quantum-mechanical definition of relative entropy
\begin{equation}
S_\mathrm{rel}\left( \Omega \Vert U' \Omega \right) = \tr_1 \left[ \rho_1 \left( \ln \rho_1 - \ln \sigma_1 \right) \right] \eqend{,}
\end{equation}
where the trace is over the first factor $\mathbb{C}^2$ of $\mathcal{H}$, and where $\rho_1$ and $\sigma_1$ are the reduced density matrices on $\mathbb{C}^2$ (traced over the second factor $\mathbb{C}^2$ of $\mathcal{H}$) of the orthogonal projectors $P_\Omega$ and $P_{U \Omega}$ onto the subspaces spanned by $\Omega$ and $U \Omega$, respectively, which read
\begin{equation}
\rho_1 = \tr_2\left( P_\Omega \right) = \frac{1}{2} \begin{pmatrix} 1 & 0 \\ 0 & 1 \end{pmatrix} = \tr_2\left( P_{U \Omega} \right) = \sigma_1 \eqend{.}
\end{equation}

Therefore, we see that it is important for a non-trivial result that unitary excitations do not have a unitary swapping partner.

\subsection{Chiral current on a light ray}

We now study the free chiral current $j$ on a light ray in 1+1 dimensions, depending only on the null coordinate $u = t-x$, and consider the positive ray with $u \geq 0$. The current can be realized as the $u$ derivative of the massless scalar field, which shows that the integral kernel of its two-point function in the Minkowski vacuum state $\omega$ is given by
\begin{splitequation}
\label{eq:twopf_current}
\omega_2(u,u') &= \frac{1}{2} \lim_{\epsilon \to 0^+} \int_0^\infty p \, \mathe^{- \mathi p (u-u'-\mathi \epsilon)} \frac{\total p}{2 \pi} \\
&= - \frac{1}{4 \pi} \lim_{\epsilon \to 0^+} \partial_u \partial_{u'} \ln (u-u'-\mathi \epsilon) \\
&= - \frac{1}{4 \pi} \partial_u \partial_{u'} \ln \abs{u-u'} + \frac{\mathi}{4} \delta'(u-u') \eqend{.}
\end{splitequation}
Here, the logarithm arose from the two-point function of the massless scalar~\cite[Eq.~(2.3b)]{schroertruong1977}, and to obtain the last expression we used the identity $\lim_{\epsilon \to 0^+} \ln(x - \mathi \epsilon) = \ln \abs{x} - \mathi \pi \Theta(-x)$. It follows that the integral kernel of the commutator function $E$ reads
\begin{splitequation}
\label{eq:commutator_current}
E(u,u') \coloneq - \mathi \Bigl[ \omega_2(u,u') - \omega_2(u',u) \Bigr] = \frac{1}{2} \delta'(u-u') \eqend{.}
\end{splitequation}
In the GNS representation of $\omega$, the vacuum vector $\Omega$ is cyclic and separating for the von Neumann algebra $\mathfrak{M}$ generated by Weyl operators $W(f) = \mathe^{\mathi j(f)}$ with real Schwartz functions $f \in \mathcal{S}_\mathbb{R}(\mathbb{R}_+)$ localized on the positive ray, and the modular operator for the pair $(\mathfrak{M}, \Omega)$ coincides with dilations along the ray~\cite{borchers1992,borchersyngvason1999}. Moreover, the commutant $\mathfrak{M}'$ is the von Neumann algebra generated by Weyl operators $W(f)$ with real Schwartz functions $f \in \mathcal{S}_\mathbb{R}(\mathbb{R}_-)$ localized on the negative ray. The Weyl relations~\eqref{eq:weyl_relation} again hold, with the commutator function given by~\eqref{eq:commutator_current}. The modular objects are again second-quantized operators on Fock space, such that
\begin{equation}
\label{eq:dilation_modular_action}
\Delta_\Omega^{\mathi t} W(f) \Delta_\Omega^{-\mathi t} = W(f_t) \quad\text{with}\quad f_t(u) = (\alpha_t f)(u) \coloneq f(\mathe^{2 \pi t} u) \ \forall t \in \mathbb{R} \eqend{.}
\end{equation}
Clearly both the positive and the negative ray are left invariant by dilations, and hence this is again an instance of a local geometric action of the modular group. Moreover, the modular conjugation $J$ acts on Weyl operators according to
\begin{equation}
\label{eq:ray_modular_conjugation}
J W(f) J = W(-f_J) \quad\text{with}\quad f_J(u) \coloneq f(-u) \eqend{,}
\end{equation}
and is extended by complex antilinearity.

We then proceed analogously to the case of the scalar field in a wedge. Consider thus a function $f \in L^1_\mathbb{R}(\mathbb{R}_+)$, and for $n \in \mathbb{N}$ the function $f_n$ defined by
\begin{equation}
\label{eq:ray_analytic_fn_def}
f_n(u) \coloneq \sqrt{ \frac{n}{\pi} } \int_{-\infty}^\infty \mathe^{- n s^2} f(\mathe^{2 \pi s} u) \total s \eqend{,}
\end{equation}
and its \emph{swapping partner}
\begin{equation}
\label{eq:ray_analytic_fn_swapping}
\left( J \alpha_\frac{\mathi}{2} f_n \right)(u) \coloneq - \sqrt{ \frac{n}{\pi} } \int_{-\infty}^\infty \mathe^{- n \left( s - \frac{\mathi}{2} \right)^2} f\left( - \mathe^{2 \pi s} u \right) \total s \eqend{.}
\end{equation}
Clearly the swapping partner is supported on the negative ray, and we again verify that nevertheless, both define the same one-particle state. For this, recall that the GNS construction establishes the one-particle Hilbert space by completing $\mathcal{S}(\mathbb{R})$ in the norm induced by the scalar product
\begin{equation}
\label{eq:ray_scalar_product}
\left( f, g \right) \coloneq \iint_{-\infty}^\infty f^*(u) \, \omega_2(u,u') g(u') \total u \total u' = \frac{1}{2} \lim_{\epsilon \to 0^+} \int_0^\infty p \, \mathe^{- p \epsilon} \left[ ( \mathcal{F} f )(p) \right]^* ( \mathcal{F} g )(p) \frac{\total p}{2 \pi}
\end{equation}
involving the two-point function $\omega_2$~\eqref{eq:twopf_current}. We therefore need to show that the Fourier transforms $( \mathcal{F} f_n )(p)$ and $\left( \mathcal{F} J \alpha_\frac{\mathi}{2} f_n \right)(p)$ agree for $p \geq 0$. We compute
\begin{splitequation}
\label{eq:ray_fourier_fn}
( \mathcal{F} f_n )(p) &= \int f_n(u) \, \mathe^{\mathi p u} \total u = \sqrt{ \frac{n}{\pi} } \int_{-\infty}^\infty \int \mathe^{- n s^2} f\left( \mathe^{2 \pi s} u \right) \, \mathe^{\mathi p u} \total u \total s \\
&= \sqrt{ \frac{n}{\pi} } \int_{-\infty}^\infty \mathe^{- n s^2 - 2 \pi s} \int f(u) \, \mathe^{\mathi p \mathe^{- 2 \pi s} u} \total u \total s \eqend{,}
\end{splitequation}
where we could interchange the integrals using Fubini's theorem since $f \in L^1_\mathbb{R}(\mathbb{R}_+)$ and $\mathe^{- n s^2 - 2 \pi s}$ is integrable in $s$ for all $n > 0$. For complex $z$, we have
\begin{equation}
\mathe^{\mathi p \mathe^{- 2 \pi z} u} = \mathe^{\mathi p u \, \mathe^{- 2 \pi \Re z} \cos(2 \pi \Im z)} \exp\left[ p u \, \mathe^{- 2 \pi \Re z} \sin(2 \pi \Im z) \right] \eqend{,}
\end{equation}
and since $u \geq 0$ we see that the integrand in~\eqref{eq:ray_fourier_fn} can for $p \geq 0$ be analytically continued to $\Im s \in \left[ - \frac{1}{2}, 0 \right]$ where $\mathe^{\mathi p \mathe^{- 2 \pi s} u}$ is bounded. Taking $\Im s = - \frac{1}{2}$, it follows that
\begin{splitequation}
\label{eq:ray_fourier_fn_final}
( \mathcal{F} f_n )(p) \bigr\rvert_{p \geq 0} &= - \sqrt{ \frac{n}{\pi} } \int_{-\infty}^\infty \mathe^{- n \left( s - \frac{\mathi}{2} \right)^2 - 2 \pi s} \int f(u) \, \mathe^{- \mathi p u \, \mathe^{- 2 \pi s}} \total u \total s \Bigr\rvert_{p \geq 0} \\
&= - \sqrt{ \frac{n}{\pi} } \int_{-\infty}^\infty \mathe^{- n \left( s - \frac{\mathi}{2} \right)^2} \int f\left( - \mathe^{2 \pi s} u \right) \, \mathe^{\mathi p u} \total u \total s \Bigr\rvert_{p \geq 0} \eqend{,}
\end{splitequation}
which is exactly the Fourier transform of the swapping partner $J \alpha_\frac{\mathi}{2} f_n$~\eqref{eq:ray_analytic_fn_swapping}, restricted to $p \geq 0$. Therefore, even though $f_n$ and $J \alpha_\frac{\mathi}{2} f_n$ are supported on opposite rays, such that the fields $j(f_n)$ and $j\left( J \alpha_\frac{\mathi}{2} f_n \right)$ are affiliated with the algebra $\mathfrak{M}$ and its commutant $\mathfrak{M}'$, respectively, the states $j(f_n) \Omega$ and $j\left( J \alpha_\frac{\mathi}{2} f_n \right) \Omega$ are the same.

We now consider the concrete example $f(u) = u \mathe^{- \alpha u}$ for $u \geq 0$ with a parameter $\alpha > 0$ and $f(u) = 0$ for $u < 0$, whose Fourier transform reads
\begin{equation}
( \mathcal{F} f )(p) = ( \alpha - \mathi p )^{-2} \eqend{.}
\end{equation}
For the family $\{ f_n \}$~\eqref{eq:ray_analytic_fn_def} we obtain using Eq.~\eqref{eq:ray_fourier_fn}
\begin{splitequation}
( \mathcal{F} f_n )(p) &= \sqrt{ \frac{n}{\pi} } \int_{-\infty}^\infty \mathe^{- n s^2 - 2 \pi s} ( \mathcal{F} f )\left( p \, \mathe^{- 2 \pi s} \right) \total s = \sqrt{ \frac{n}{\pi} } \int_{-\infty}^\infty \frac{\mathe^{- n s^2}}{\left( \alpha \, \mathe^{\pi s} - \mathi p \, \mathe^{- \pi s} \right)^2} \total s \\
&= \frac{1}{\sqrt{\pi}} \int_0^\infty \mathe^{- s^2} \frac{4 \mathi \alpha p - 2 (\alpha^2 - p^2) \cosh\left( \frac{\pi s}{\sqrt{n}} \right)}{\left[ \mathi ( \alpha^2 - p^2 ) + 2 \alpha p \cosh\left( \frac{\pi s}{\sqrt{n}} \right) \right]^2} \total s \eqend{,}
\end{splitequation}
where in the last equality we split the integral at $s = 0$, and performed the change of variables $s \to \pm s/\sqrt{n}$. While the integral cannot be computed exactly, it is possible to obtain an asymptotic expansion for large $n$ by expanding the integrand, which is uniformly bounded in $n$. Namely, we obtain
\begin{equation}
h(s) \coloneq \frac{4 \mathi \alpha p - 2 (\alpha^2 - p^2) \cosh\left( \frac{\pi s}{\sqrt{n}} \right)}{\left[ \mathi ( \alpha^2 - p^2 ) + 2 \alpha p \cosh\left( \frac{\pi s}{\sqrt{n}} \right) \right]^2} = \frac{2}{(\alpha - \mathi p)^2} \left[ 1 + \frac{2 \pi^2 s^2}{n} \frac{\alpha^2 - p^2 + 4 \mathi \alpha p}{(\alpha - \mathi p)^2} + \bigo{\frac{1}{n^2}} \right]
\end{equation}
and the remainder can be bounded by $\frac{s^4}{4!} \abs{ h^{(4)}(\sigma) }$ for some $\sigma \geq 0$. The integrals can now be computed, and we obtain
\begin{equation}
\label{eq:ray_ffn}
( \mathcal{F} f_n )(p) = \frac{1}{(\alpha - \mathi p)^2} \left[ 1 + \frac{\pi^2}{n} \frac{\alpha^2 - p^2 + 4 \mathi \alpha p}{(\alpha - \mathi p)^2} + \frac{C(n,p)}{n^2} \right] \eqend{,}
\end{equation}
where $C(n,p)$ is a function that can be bounded by a constant.\footnote{Long but straightforward computations result in $\abs{C(n,p)} \leq 10^3$ for $n \geq 10^2$.} It follows that
\begin{splitequation}
\label{eq:ray_jfn_norm}
\norm{ j(f_n) \Omega }^2 &= \frac{1}{2} \lim_{\epsilon \to 0^+} \int_0^\infty p \, \mathe^{- p \epsilon} \abs{ ( \mathcal{F} f_n )(p) }^2 \frac{\total p}{2 \pi} \\
&= \frac{1}{8 \pi \alpha^2} - \frac{\pi}{4 n \alpha^2} + \frac{C'(n,\alpha)}{n^2 \alpha^2} \eqend{,}
\end{splitequation}
where $C'(n,\alpha)$ is another function that can be bounded by a constant.\footnote{Straightforward computations result in $\abs{C'(n,\alpha)} \leq 150$ for $n \geq 10^2$.} We thus obtain the normalized state
\begin{equation}
\label{eq:ray_psifn_norm}
\psi_n \coloneq c_n j(f_n) \Omega = j(c_n f_n) \Omega \eqend{,} \quad c_n = \left[ \frac{1}{8 \pi \alpha^2} - \frac{\pi}{4 n \alpha^2} + \frac{C'(n,\alpha)}{n^2 \alpha^2} \right]^{-\frac{1}{2}} \eqend{,}
\end{equation}
and can bound the relative entropy between $\psi_n$ and $\Omega$ using Cor.~\ref{corr:bound2}, or the one between $\psi \coloneq c_\infty \, j(f) \Omega$ and $\Omega$ using Cor.~\ref{corr:bound3}. Note that $\psi_n$ converges in norm to $\psi$, since
\begin{equation}
\norm{ j(f_n-f) \Omega }^2 = \frac{1}{2} \int_0^\infty p \abs{ ( \mathcal{F} (f_n - f) )(p) }^2 \frac{\total p}{2 \pi} \eqend{,}
\end{equation}
and from~\eqref{eq:ray_ffn} we see that
\begin{equation}
( \mathcal{F} (f_n-f) )(p) = \frac{1}{(\alpha - \mathi p)^2} \left[ \frac{\pi^2}{n} \frac{\alpha^2 - p^2 + 4 \mathi \alpha p}{(\alpha - \mathi p)^2} + \frac{C(n,p)}{n^2} \right] \eqend{,}
\end{equation}
is bounded by an integrable function of $p$ and vanishes pointwise as $n \to \infty$.

Since $j(f_n)$ is an analytic element affiliated to $\mathfrak{M}$, we can simply compute the analytic continuation of its modular flow by analytically continuing the integrand, and obtain $\sigma_t\left( j(f_n) \right) = j\left( \alpha_t f_n \right)$ with
\begin{splitequation}
(\mathcal{F} \alpha_t f_n)(p) &= \sqrt{ \frac{n}{\pi} } \int \int_{-\infty}^\infty \mathe^{- n (s-t)^2} f\left( \mathe^{2 \pi s} u \right) \total s \, \mathe^{\mathi p u} \total u \\
&= \sqrt{ \frac{n}{\pi} } \int_{-\infty}^\infty \mathe^{- n (s-t)^2 - 2 \pi s} ( \mathcal{F} f )\left( p \, \mathe^{- 2 \pi s} \right) \total s \\
&= \frac{1}{\sqrt{\pi}} \int_0^\infty \mathe^{- \left( s - \frac{t}{\sqrt{n}} \right)^2} \frac{4 \mathi \alpha p - 2 (\alpha^2 - p^2) \cosh\left( \frac{\pi s}{\sqrt{n}} \right)}{\left[ \mathi ( \alpha^2 - p^2 ) + 2 \alpha p \cosh\left( \frac{\pi s}{\sqrt{n}} \right) \right]^2} \total s \eqend{.}
\end{splitequation}
Expanding again for large $n$, the same steps as for $t = 0$ yield
\begin{equation}
\label{eq:ray_fourier_alphat_fn}
(\mathcal{F} \alpha_t f_n)(p) = \frac{1}{(\alpha - \mathi p)^2} \left[ 1 + \frac{2 t}{\sqrt{\pi n}} + \frac{\pi^2}{n} \frac{\alpha^2 - p^2 + 4 \mathi \alpha p}{(\alpha - \mathi p)^2} + \frac{C''(n,p)}{n^\frac{3}{2}} \right] \eqend{,}
\end{equation}
where $C''(n,p)$ is another function that can be bounded by a constant, assuming $\abs{t} \leq 1$. Since $f$ is a real function, we have $\left[ \sigma_t\left( j(f_n) \right) \right]^* = \sigma_{t^*}\left( j(f_n) \right) = j\left( \alpha_{t^*} f_n \right)$, and because the chiral current is free, Wick's theorem results in
\begin{splitequation}
\norm{ \sigma_\frac{\mathi}{4}\left( j(c_n f_n) \right) \left[ \sigma_\frac{3 \mathi}{4}\left( j(c_n f_n) \right) \right]^* \Omega }^2 &= c_n^4 \left( \Omega, j\left( \alpha_\frac{3 \mathi}{4} f_n \right) j\left( \alpha_{-\frac{\mathi}{4}} f_n \right) j\left( \alpha_\frac{\mathi}{4} f_n \right) j\left( \alpha_{-\frac{3 \mathi}{4}} f_n \right) \Omega \right) \\
&= c_n^4 \left( \alpha_{-\frac{3 \mathi}{4}} f_n, \alpha_{-\frac{\mathi}{4}} f_n \right) \left( \alpha_{-\frac{\mathi}{4}} f_n, \alpha_{-\frac{3 \mathi}{4}} f_n \right) \\
&\quad+ c_n^4 \left( \alpha_{-\frac{3 \mathi}{4}} f_n, \alpha_\frac{\mathi}{4} f_n \right) \left( \alpha_\frac{\mathi}{4} f_n, \alpha_{-\frac{3 \mathi}{4}} f_n \right) \\
&\quad+ c_n^4 \left( \alpha_{-\frac{3 \mathi}{4}} f_n, \alpha_{-\frac{3 \mathi}{4}} f_n \right) \left( \alpha_\frac{\mathi}{4} f_n, \alpha_\frac{\mathi}{4} f_n \right)
\end{splitequation}
with the scalar product~\eqref{eq:ray_scalar_product} of the GNS construction. From the explicit expression~\eqref{eq:ray_fourier_alphat_fn} of the Fourier transform of the function $\alpha_t f_n$, we see that each term has a finite limit as $n \to \infty$, which is in fact independent of $t$. We thus obtain
\begin{equation}
\lim_{n \to \infty} \norm{ \sigma_\frac{\mathi}{4}\left( j(c_n f_n) \right) \left[ \sigma_\frac{3 \mathi}{4}\left( j(c_n f_n) \right) \right]^* \Omega }^2 = 3 \lim_{n \to \infty} c_n^4 \norm{ j(f_n) \Omega }^4 = 3 \lim_{n \to \infty} \norm{ \psi_n }^4 = 3 \eqend{,}
\end{equation}
since $c_n$ was determined in such a way that $\norm{ \psi_n } = 1$~\eqref{eq:ray_psifn_norm}. Therefore, using Cor.~\ref{corr:bound3} we obtain the bound
\begin{equation}
\label{eq:chiral_srel_bound}
0 \leq S_\mathrm{rel}\left( \Omega \Vert \psi \right) \leq 2 \ln 3 \eqend{.}
\end{equation}
From the derivation, we see that the same bound holds in fact for a large class of functions $g$, namely for real smooth functions $g$ of compact support in $(0,\infty)$. In this case, the Fourier transform $(\mathcal{F} g)(p)$ decays faster than $\frac{1}{p}$ for large $p$, and all integrals that appear are finite and can be performed in any order thanks to Fubini's theorem. Therefore, the relative entropy between the vacuum and $j(g) \Omega$ for all of these functions is finite. Since they are dense in the one-particle Hilbert space, we obtain the uniform bound~\eqref{eq:chiral_srel_bound} on a dense subset.

\section{Discussion and outlook}

We have shown how one can use the convexity of non-commutative $L^p$ norms to bound the relative entropy between the faithful state $\omega$ on a von Neumann algebra $\mathfrak{M}$ and an excitation $\psi = \omega \circ \operatorname{ad}_{b'}$ thereof, where $b'$ is an arbitrary operator affiliated to the commutant $\mathfrak{M}'$. The bound arises from the monotonicity of the Araki--Masuda divergences defined via the non-commutative $L^p$ norms and the fact that they are sandwiched between the Petz--Rénxi relative entropies, which was first proven in Refs.~\cite{bertascholztomamichel2018,jencova2018}. We have shown that for $p = 4$ and $p = \infty$ the non-commutative $L^p$ norms can be computed without needing to know the relative modular operator, and instead can be expressed using only the non-relative modular operator associated to $\omega$. We have then employed elements $a \in \mathfrak{M}$ analytic with respect to the modular flow and their swapping partners to translate our bounds into bounds for excitations of the form $\omega \circ \operatorname{ad}_a$ with analytic $a$, and finally used the lower semicontinuity of the relative entropy to obtain bounds for general excitations from the algebra.

As an example, we computed swapping partners in free scalar field theory in a wedge and for chiral currents $j$ on a light ray, and in the latter case derived the uniform bound \eqref{eq:chiral_srel_bound} on the relative entropy between the vacuum $\Omega$ and a dense subset of normalized single-particle states $\psi = j(g) \Omega$. Using the known behavior of relative entropy under rescalings~\cite[Thm.~3.6~(3)]{araki1977}, this translates into the bound
\begin{splitequation}
S_\mathrm{rel}\left( \Omega \Vert j(g) \Omega \right) &= \norm{ j(g) \Omega }^2 S_\mathrm{rel}\left( \Omega \middle\Vert \frac{j(g) \Omega}{\norm{ j(g) \Omega }} \right) + \norm{ j(g) \Omega }^2 \ln \norm{ j(g) \Omega }^2 \\
&\leq 2 \norm{ j(g) \Omega }^2 \ln \Bigl( 3 \norm{ j(g) \Omega } \Bigr)
\end{splitequation}
for non-normalized excitations $j(g) \Omega$ with an arbitrary real smooth function $g \in L^1(\mathbb{R}_+)$ vanishing for $u \leq 0$. This bound is of a qualitatively different form than the Bekenstein bound derived in~\cite{longo2025}, which depends on the energy of the state and not only on its norm. It is thus conceivable that there are families of states for which our bound remains finite, while the one of~\cite{longo2025} diverges, and vice versa.

It will be interesting to derive a bound on the relative entropy also for the case where the excitation is obtained in a more general way, not by an element affiliated to the algebra; for the chiral current this could be a state $j(g) \Omega$ with $g$ having support on both the positive and negative light rays. A first step towards deriving a corresponding bound is the understanding of generalized swapping partners for such states, namely whether there exist $a$ affiliated to $\mathfrak{M}$ (or $b'$ affiliated to $\mathfrak{M}'$) such that the states $j(g) \Omega$ and $a \Omega$ (or $b' \Omega$) are the same. In the case of the Bekenstein bound, the corresponding generalization~\cite{hollandslongo2025} results in a bound that deteriorates if a large part of the excitation comes from the commutant; if such a deterioration also arises in our case will strongly depend on properties of the generalized swapping partner. Moreover, it will be informative to apply our bounds to other theories, in particular for fermions where relative entropy has been studied in various situations (see the recent Refs.~\cite{longoxu2018,friesreyes2019,xu2023,galandamuchverch2023,finstermuch2025} and references therein).

\begin{acknowledgments}
This work has been partially funded by the Deutsche Forschungsgemeinschaft (DFG, German Research Foundation) --- project no. 396692871 within the Emmy Noether grant CA1850/1-1. L.S. acknowledges financial support by Italian Ministry of University and Research through the grant PRIN 2022ZE8SC4. M.B.F. thanks Gandalf Lechner and Ricardo Correa da Silva for discussions, and both authors thank Mario Berta and Marco Tomamichel for an email exchange and for pointing out Refs.~\cite{jencova2018,jencova2021}.
\end{acknowledgments}

\appendix

\subsection*{Conflict of interest statement}

On behalf of all authors, the corresponding author states that there is no conflict of interest.

\subsection*{Data availability statement}

This manuscript has no associated data.

\bibliography{literature}

\end{document}